\documentclass{article}
\usepackage[a4paper, margin=1in]{geometry}
\usepackage{graphicx}
\usepackage{array}
\usepackage{amssymb,amsmath,amsfonts,amsthm}
\usepackage{bm}
\usepackage{xcolor}
\usepackage{subcaption}
\usepackage{blindtext}
\usepackage{url}
\usepackage{tikz}
\usetikzlibrary{positioning,arrows.meta}

\definecolor{arrowblue}{RGB}{98,145,224}

\newcommand\ImageNode[3][]{
  \node[draw=arrowblue!80!black,line width=1pt,#1] (#2) {\includegraphics[width=0.25\textwidth]{#3}};
}

\newtheorem{definition}{Definition}
\newtheorem{theorem}{Theorem}
\newtheorem{lemma}{Lemma}

\usepackage{todonotes}

\begin{document}

\title{Nonparametric Inference Framework for Time-dependent Epidemic Models}

\author{Son Luu$^1$ \qquad Edward Susko$^1$ \qquad Lam Si Tung Ho$^1$\\[15pt]
$^1$Dalhousie University, Canada
}

\date{}

\maketitle

\begin{abstract}
Compartmental models, especially the  Susceptible-Infected-Removed (SIR) model, have long been used to understand the behaviour of various diseases. Allowing parameters, such as the transmission rate, to be time-dependent functions makes it possible to adjust for and make inferences about changes in the process due to mitigation strategies or evolutionary changes of the infectious agent. In this article, we attempt to build a nonparametric inference framework for stochastic SIR models with time dependent infection rate.
The framework includes three main steps: likelihood approximation, parameter estimation and confidence interval construction. The likelihood function of the stochastic SIR model, which is often intractable, can be approximated using methods such as diffusion approximation or tau leaping. The infection rate is modelled by a B-spline basis whose knot location and number of knots are determined by a fast knot placement method followed by a criterion-based model selection procedure. Finally, a point-wise confidence interval is built using a parametric bootstrap procedure.
The performance of the framework is observed through various settings for different epidemic patterns. The model is then applied to the Ontario COVID-19 data across multiple waves.
\end{abstract}

\section{Introduction}
Compartmental models are a type of mathematical model used to study the spread of infectious diseases through a population. The basic idea of a compartmental model is to divide the population into different compartments based on their disease status, such as the famous Susceptible-Infected-Removed (SIR) model \cite{Kermack1927}. Each compartment represents a group of individuals with the same disease status, and the model tracks the flow of individuals between compartments over time.
These models have been used to study a wide range of infectious diseases, including influenza \cite{Osthus2017, saha2023spade4, yang2015forecasting}, HIV/AIDS \cite{Blum2010, clemenccon2008stochastic, kuhnert2014simultaneous}, plague \cite{ho2018birth, raggett1982stochastic, whittles2016epidemiological}, Ebola \cite{berge2017simple, ho2018direct, saha2023spade4}, and COVID-19 \cite{Roda2020, saha2023spade4, watson2021pandemic} and have been particularly useful for understanding the dynamics of epidemics, including the timing and size of outbreaks, as well as the impact of various control measures.

Many standard epidemic models assume that the epidemic parameters, such as the transmission rate and the recovery rate, are constant over time. In reality, the parameters of an epidemic can change over time due to various factors, such as changes in the behavior of the population, the implementation of interventions, and the emergence of new variants of the pathogen. Therefore, there is a need for time-dependent epidemic models that can capture the dynamics of these changing parameters. 

For inference and prediction, there are two main types of compartmental models: deterministic and stochastic. In the deterministic model, the epidemic dynamics are described by a set of differential equations and the model parameters are often obtained by solving a least square problem. Some works have implemented this model type with time dependent rates \cite{Chen2020,Smirnova2019}. Deterministic models tend to work well for large populations and are computationally efficient. Stochastic models, on the other hand, take into account the random variation in disease transmission within the population. Allowing stochastic variation is particularly valuable when population sizes are not very large because the dynamics can be quite different. Epidemics that are expected to infect the entire population in deterministic models can end prematurely in stochastic models and total numbers of infections can show substantial variation. For these models, the number of individuals in each compartment is often assumed to follow a Markov process. Unfortunately, exact likelihood computation for stochastic compartmental models are typically intractable or time consuming so likelihood based methods tends to use approximation methods such as diffusion approximation \cite{Britton2019,Dargatz2006}. Because of the additional complexity required to fit stochastic models, they rarely incorporate time dependent rates into a stochastic model; but see \cite{bu2024stochastic, huang2024detecting}.

With that in mind, this article explores a nonparametric inference framework for stochastic compartmental models with time dependent rates, specifically the SIR model with time dependent infection rate. There are two main underlying ideas: using a spline basis to estimate the true rates as a function of time and using simpler processes to approximate the often intractable likelihood function. For inference, a fast spline knot placement method \cite{Yeh2020} is employed and assisted by a moving average rate estimate. Then various aspects of the model are examined in a simulation study including approximation type, model selection procedure and numerical considerations. Finally, the model is applied to estimate COVID-19 patterns in Ontario over multiple waves. 

The rest of the paper is structured as follows. Section \ref{background} provides the necessary background; Section \ref{lik_approx} describes the basis for likelihood approximation; Section \ref{RS} introduces the regression spline (RS) framework and how it perform parameter estimation; Section \ref{CI} discusses confidence interval construction; Section \ref{sim_study} discusses the simulation study results; and Section \ref{application} applies the proposed framework to the Ontario COVID-19 data.

\section{Background}\label{background}
This section will go over the definitions and properties of the stochastic processes involved in the model construction, B-spline basis \cite{Wasserman2006}, Wasserstein distance \cite{villani2009optimal} and the parametric bootstrap procedure \cite{efron1994introduction}.

\subsection{Stochastic SIR model}
In this model, the population with an on-going disease is divided into three compartments: susceptible ($S$) for those who are not yet infected, infected ($I$) for those who are infected, and removed ($R$) for those who recovered or died from the disease. As illustrated in figure \ref{FigSIR}, there are two types of movements for an individual in the population: getting infected by the disease ($S\to I$) and recovering (or dying) from the disease ($I\to R$).

For a closed population of $N$ individuals, let $S(t), I(t)$ and $R(t)=N-S(t)-I(t)$ be the number of susceptible, infected and removed individuals at time $t$, respectively. Then at time $t$, individuals move from $S$ to $I$ with rate $\beta(t)S(t)I(t)$ and from $I$ to $R$ with rate $\gamma(t)I(t)$. Here $\beta(t)$ and $\gamma(t)$ are the infection and recovery rates at time $t$, respectively. 

\begin{figure}[h]
    \centering
    \includegraphics[width=0.8\textwidth]{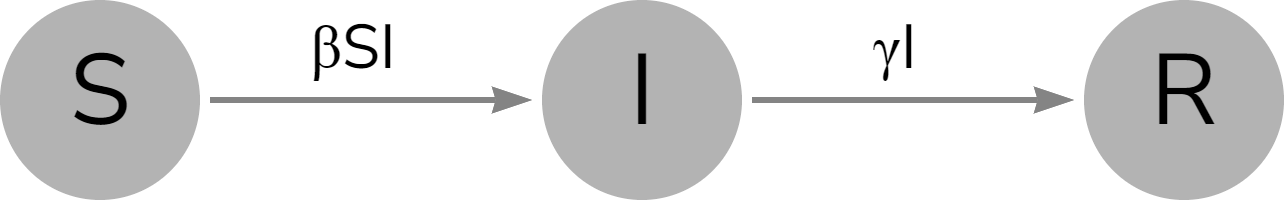}
    \caption{Graph representation of the SIR model.}
    \label{FigSIR}
\end{figure}

In the stochastic SIR model, these movements can be formally described by a bivariate continuous-time Markov process as follows

\begin{definition}\label{def1}
The stochastic SIR model assumes that $X(t) = (S(t), I(t))$ is a bivariate continuous-time Markov process satisfying
\begin{align*}
    X(t) = X(0) + \begin{pmatrix}-1\\ 1 \end{pmatrix} Pois_1\left(\int_0^t\beta(s)\dfrac{S(s)I(s)}{N}ds\right) + \begin{pmatrix}0\\ -1 \end{pmatrix} Pois_2\left(\int_0^t\gamma I(s)ds\right)
\end{align*}
where $Pois_1,Pois_2$ are independent standard Poisson processes.
\end{definition}

For this paper, the focus is on the SIR model where the recovery rate $\gamma$ is constant.

\subsection{Diffusion process}
Diffusion processes are continuous-time stochastic processes whose sample paths are continuous. A simple example of this is the Brownian motion. These processes are often described by a stochastic differential equation (SDE) as follows

\begin{definition}
A diffusion process $Z(t)$ is a continuous-time Markov process that satisfies the Ito SDE
\begin{align*}
    dZ(t) = A(t,Z(t))dt + L(t,Z(t))dB(t)
\end{align*}
where $B(t)$ is a multivariate Brownian motion, $A(t,z)$ and $L(t,z)$ are called the drift vector and diffusion matrix, respectively.
\end{definition}
A simple interpretation of this is that the drift vector controls the mean of the process and the diffusion matrix controls the variance. In later sections, the diffusion process will be used to approximate the stochastic SIR model.

\subsection{B-spline}
B-spline a well known method of curve fitting using piece-wise polynomials.

\begin{definition}
Let $t_0<t_1<t_2<\ldots<t_k<t_{k+1}$ be $k$ points (known as knots) in an interval $(t_0,t_{k+1})$. A B-spline $f$ of order $d+1$ is a piece-wise degree $d$ polynomial defined by the formula
\begin{align*}
    f(t) = \sum_{i=1}^{k+d+1}c_i\varphi_{i,d}(t)
\end{align*}
where $\varphi_{i,d}$ are degree $d$ polynomials in $(t_{i-d-1},t_i)$\footnote{If $j<0$, set $t_j=t_0$. If $j>k+1$, set $t_j=t_{k+1}$ } called the basis functions and $c_i$ are the corresponding coefficients.
\end{definition}

For the construction of the basis functions, set $\tau_1=\ldots=\tau_{d+1}=t_0$, $\tau_{i+d+1}=t_i$ for $i=1,\ldots,k$ and $t_{k+1}=\tau_{k+d+2}=\ldots=\tau_{k+2d+2}$. Then
\begin{align*}
    &\varphi_{i,0}(t) = \begin{cases}
    1 \quad \text{if } t\in [\tau_{i},\tau_{i+1})\\
    0 \quad \text{otherwise}
    \end{cases},\\
    &\varphi_{i,d}(t) = \dfrac{t-\tau_i}{\tau_{i+d}-\tau_i}\varphi_{i,d-1}(t)+\dfrac{\tau_{i+d+1}-t}{\tau_{i+d+1}-\tau_{i+1}}\varphi_{i+1,d-1}(t).
\end{align*}

A note worthy feature of B-spline is that their basis functions have compact support which can speed up calculations \cite{Wasserman2006}.

\subsection{Wasserstein distance}
The Wasserstein distance is commonly used for measuring the difference between two distributions.
\begin{definition}
Let $U,V$ be two $\mathbb{R}^d$-valued random variables. The Wasserstein-1 distance between them is defined as
\begin{align*}
    W_{1}(U,V) = \inf E(\|U-V\|)
\end{align*}
where the infimum is over all possible coupling of $U$ and $V$, i.e. all ways of jointly defining the two variables while respecting their marginal distribution. Note that the norm $\|\cdot\|$ is simply the Euclidean norm.
\end{definition}

To measure the difference between two stochastic processes, we modify the above definition as follow
\begin{definition}
Let $U(t),V(t)$ be two $\mathbb{R}^d$-valued stochastic processes on the interval $[0,T]$. Then the Wasserstein-1 distance between them is
\begin{align*}
    W_{1,T}(U,V)= \inf E(\|U-V\|_T)
\end{align*}
where $\|X\|_T=\sup\limits_{t\in [0,T]} \|X(t)\|$ and the infimum is over all possible coupling of $U(t)$ and $V(t)$, i.e. all ways of jointly defining the two processes while respecting their marginal distribution.
\end{definition}

Next we have a lemma about Wasserstein distance and the point-wise law of processes 
\begin{lemma}\label{lem_Wasserstein}
If the stochastic process sequence $U_n(t)$ and the stochastic process $V(t)$ on $[0,T]$ satisfy
\begin{align*}
    W_{1,T}(U_n,V)\xrightarrow{n\to\infty} 0.
\end{align*}
Then for all $t\in [0,T]$ we have
\begin{align*}
    U_n(t) \xrightarrow{d} V(t)
\end{align*}
in other words, $U_n(t)$ converges in law to $V(t)$.
\end{lemma}
The proof of this lemma can be found in the appendix.

\subsection{Parametric Bootstrap and Bootstrap confidence intervals}

Consider an estimation problem where the quantity of interest is $\theta$ and the data generating distribution is $F_\theta$. Now assume that we have a procedure to obtain an estimate $\hat{\theta}$ of $\theta$. The parametric bootstrap is a method to estimate the distribution of $\hat{\theta}$. We accomplish that goal by performing the following steps:
\begin{enumerate}
    \item[1. ] Generate a sample from the approximate distribution $F_{\hat{\theta}}$.
    \item[2. ] Obtain an estimate $\hat{\theta}^*$ of $\hat{\theta}$.
    \item[3. ] Repeat steps 1 and 2 $B$ times to get $\hat{\theta}^*_1,\ldots,\hat{\theta}^*_B$.
\end{enumerate}

With the sample of estimates in step 3, we can estimate various aspects of $\hat{\theta}$ such as the bias, variance and confidence interval.

\subsubsection{Bootstrap confidence intervals}
In this subsection, we define all types of bootstrap confidence intervals that are utilized in later sections.
\textbf{Pivotal Interval}
The $1-\alpha$ bootstrap pivotal interval is defined as
\begin{align*}
    CI_{pivotal} = \left(2\hat{\theta}-\hat{\theta}^*_{(B(1-\alpha/2))},\ 2\hat{\theta}-\hat{\theta}^*_{(B\alpha/2)}\right)
\end{align*}
where $\hat{\theta}^*_{(B\alpha)}$ denotes the $100\alpha^{th}$ quantile of $\hat{\theta}^*$. This interval works by utilizing the distribution of the pivot $\hat{\theta}-\theta$.

\textbf{Normal Interval}
The $1-\alpha$ bootstrap normal interval is defined as
\begin{align*}
    CI_{normal} = \left(\hat{\theta}-z_{\alpha/2}s_b,\ \hat{\theta}+z_{\alpha/2}s_b\right)
\end{align*}
where $z_{\alpha/2}$ is the $1-\alpha/2$ quantile of the standard normal and $s_b$ is the bootstrap estimate of the standard error. This interval works under the assumption that the distribution of $\hat{\theta}$ is close to normal, i.e. $\hat{\theta}\sim N(\theta,s^2)$.

\textbf{Percentile Interval}
The $1-\alpha$ bootstrap percentile interval is defined as
\begin{align*}
    CI_{percentile} = \left(\hat{\theta}^*_{(B\alpha/2)},\ \hat{\theta}^*_{(B(1-\alpha/2))}\right).
\end{align*}
This interval works under the assumption that there exists a monotonic increasing transformation $\rho$ such that $\rho(\hat{\theta})\sim N(\rho(\theta),c^2)$.

\subsubsection{Bias correction for Bootstrap confidence intervals}
In many inference problems, especially nonparametric ones, there will be a certain amount of bias
\begin{align*}
    b = E\hat{\theta}-\theta.
\end{align*}
To account for these biases, we will have to make some adjustments to the bootstrap confidence intervals. These adjustments often involve subtracting the bias, which is estimated by the bootstrap bias
\begin{align*}
    \hat{b}^*=\dfrac{1}{B}\sum_{i=1}^B\hat{\theta}^*_i-\hat{\theta}.
\end{align*}

\textbf{Pivotal Interval}
The appearance of a bias does not affect this interval's main assumption, that is the distribution of the pivot $\hat{\theta}-\theta$ is close to the distribution of $\hat{\theta}^*-\hat{\theta}$. Therefore, the pivotal interval has already accounted for bias correction.

\textbf{Normal Interval}
When there is a bias term, the assumption for this interval becomes
\begin{align*}
    \hat{\theta}\sim N(\theta+b,s^2).
\end{align*}
Then, the bias corrected confidence interval will be
\begin{align*}
    CI_{\text{corrected normal}} = \left(\hat{\theta}-\hat{b}^*-z_{\alpha/2}s_b,\ \hat{\theta}-\hat{b}^*+z_{\alpha/2}s_b\right).
\end{align*}

\textbf{Percentile Interval}
In this case, the main assumption is reinterpreted as
\begin{align*}
    \exists\ \rho\text{ monotonic increasing: } \rho(\hat{\theta}-b)\sim N(\rho(\theta),c^2).
\end{align*}
Then the bias corrected percentile confidence interval can be derived as follows.
\begin{align}
\begin{alignedat}{3}
    1-\alpha &= P(\rho(\theta)-z_{\alpha/2}\le \rho(\hat{\theta}-b) \le \rho(\theta)+z_{\alpha/2})\\
    &= P(\rho(\hat{\theta}-b)-z_{\alpha/2}\le \rho(\theta) \le \rho(\hat{\theta}-b)+z_{\alpha/2})\\
    &= P(\rho(\hat{\theta}-b)_{\alpha/2}\le \rho(\theta) \le \rho(\hat{\theta}-b)_{1-\alpha/2}). \label{eq2.29}
\end{alignedat}
\end{align}
Now since $\rho$ is monotonic increasing, it preserves quantiles so $\rho(\hat{\theta}-b)_{\alpha}=\rho(\hat{\theta}_{\alpha}-b)$ for all $\alpha$. Therefore \eqref{eq2.29} becomes
\begin{align}
\begin{alignedat}{3}
    1-\alpha &= P(\rho(\hat{\theta}_{\alpha/2}-b)\le \rho(\theta) \le \rho(\hat{\theta}_{1-\alpha/2}-b))\\
    &=P(\hat{\theta}_{\alpha/2}-b\le \theta \le \hat{\theta}_{1-\alpha/2}-b). \label{eq2.31}
\end{alignedat}
\end{align}
Then the quantiles of $\hat{\theta}$ are estimated by the bootstrap sample while keeping in mind that there is a bias term
\begin{align}
    \forall\ \alpha:\quad \hat{\theta}_{\alpha}\approx \hat{\theta}^*_{B\alpha}-\hat{b}^*. \label{eq2.32}
\end{align}
Plugging \eqref{eq2.32} into \eqref{eq2.31} and replacing $b$ with $\hat{b}^*$, we get
\begin{align*}
    1-\alpha = P(\hat{\theta}^*_{(B\alpha/2)}-2\hat{b}^*\le \theta \le \hat{\theta}^*_{(B(1-\alpha/2))}-2\hat{b}^*).
\end{align*}
Hence, the formula for the bias corrected percentile confidence interval is
\begin{align*}
    CI_{\text{corrected percentile}} = \left(\hat{\theta}^*_{(B\alpha/2)}-2\hat{b}^*,\ \hat{\theta}^*_{(B(1-\alpha/2))}-2\hat{b}^*\right).
\end{align*}

\section{Likelihood Approximation}\label{lik_approx}
Consider the stochastic SIR model, as defined in section \ref{background}, with infection rate function $\beta(t)$ and constant recovery rate $\gamma$. Our goal is to estimate both $\beta(t)$ and $\gamma$ using discretely observed data of the number of susceptible and infected individuals. To this end, we use a set of parameters $\theta=\{\theta_j\}_{j=1}^p$ where $\gamma=\theta_1$ and the other parameters are used to model $\beta(t)$. With our model set up, the likelihood function is as follow
\begin{align}\label{3.4}
L_X(\theta)=\prod_{i=1}^{M-1}P_{\theta,t_i,t_{i+1}}(X(t_{i+1})|X(t_i)).
\end{align}
The biggest problem here is computing or approximating the transition probabilities in \eqref{3.4}. Therefore, a stochastic process whose transition probabilities can be tractably approximated is used to approximate our SIR model. Note that analysis here is of the data $X(t_2),\dots,X(t_M)$, conditional on $X(t_1).$

\subsection{Diffusion approximation}
In this section, the diffusion process used to approximate our SIR is presented along with the convergence results. Now set $x(t) = (s(t),j(t)) = X(t)/N = (S(t)/N,I(t)/N)$. This rescaled process represents the proportion of susceptible and infected in the population. Using the new description, the state space of $x(t)$ can be viewed as ``continuous" for large $N$, making the approximation to a diffusion process more natural.
Next, consider the diffusion process $z(t)=(s_z(t),j_z(t))$ as follow 
\begin{equation}\label{3.5}
\begin{alignedat}{3}
&ds_z(t) = -\beta(t)s_z(t)j_z(t)dt + \sqrt{\dfrac{\beta(t)s_z(t)j_z(t)}{N}}dB_1(t)\\
&dj_z(t) = (\beta(t)s_z(t)j_z(t)-\gamma j_z(t))dt - \sqrt{\dfrac{\beta(t)s_z(t)j_z(t)}{N}}dB_1(t) + \sqrt{\dfrac{\gamma j_z(t)}{N}}dB_2(t)
\end{alignedat}
\end{equation}
where $B_1(t), B_2(t)$ are independent Brownian motions. This process is similar to the deterministic version of the SIR model with added white noise accounting for stochasticity in each compartment. Rewriting \eqref{3.5} to its matrix form gives us
\begin{align}\label{3.6}
dz(t) = A(t,z(t))dt + L(t,z(t))dB(t)
\end{align}
where $B(t)$ is a bivariate Brownian motion and
\begin{align*}
    A(t,z(t)) = \begin{pmatrix} -\beta(t)s_z(t)j_z(t) \\ \beta(t)s_z(t)j_z(t)-\gamma j_z(t) \end{pmatrix},\quad
    L(t,z(t)) = \dfrac{1}{\sqrt{N}}\begin{pmatrix} \sqrt{\beta(t)s_z(t)j_z(t)} & 0 \\ -\sqrt{\beta(t)s_z(t)j_z(t)} & \sqrt{\gamma j_z(t)} \end{pmatrix}.
\end{align*}
Next, we have the following theorem
\begin{theorem}\label{theo3.1}
Let $[0,T]$ be the time interval of the data. Then we have
\begin{align*}
    \sqrt{N}W_{1,T}(x,z)\xrightarrow{N\to\infty}0
\end{align*}
or in other words, $W_{1,T}(x,z)=o(1/\sqrt{N})$.
\end{theorem}
The proof for a more generalized version of theorem \ref{theo3.1}, where $x(t)$ is a general compartmental model, can be found in \cite{Britton2019}. The main idea is to prove that both $x(t)$ and $z(t)$ converge in Wasserstein distance to the same process. With this, we can, for sufficiently large $N$, use lemma \ref{lem_Wasserstein} to replace the likelihood function in \eqref{3.4} with
\begin{align}\label{3.9}
L_z(\theta)=\prod_{i=1}^Mp_{\theta,t_i,t_{i+1}}(z(t_{i+1})=x(t_{i+1})|z(t_i)=x(t_i)).
\end{align}

\subsubsection{Likelihood computation for diffusion processes}
We will now look into methods to compute the conditional densities in \eqref{3.9}. If the SDE in \eqref{3.6} is explicitly solvable, then the likelihood function can be exactly computed. For example, assuming that the solution can be written as
\begin{align}\label{eq:exact}
    z(t) = z(0) + D(t,\theta)+E(t,\theta)B(t)
\end{align}
where $D,E$ are functions of appropriate dimensions. Then
\begin{align}
    z(t_{i+1})|z(t_i) \sim N(z(0) + D(t_{i+1},\theta)+E(t_{i+1},\theta)B(t_i), E(t_{i+1},\theta)\Delta t_iI_2)
\end{align}
where $\Delta t_i=t_{i+1}-t_i$, $I_2$ the rank 2 identity matrix and $B(t_i)$ can be determined by $z(t_i)$ and \eqref{eq:exact}. With this, we can get a closed form expression for $p_{\theta}(z(t_{i+1})|z(t_i))$.

However, the SDE in \eqref{3.6} is not explicitly solvable in general and therefore requires a different approach. The method I settled on involves the simple Euler-Maruyama approximation $\tilde{z}^{(k)}(t)$ of $z(t)$. For all observed time $t_i$, let 
\begin{align}\label{3.12}
\begin{alignedat}{3}
\tau_{ir}&=t_i+r\dfrac{\Delta t_i}{k}=t_i+r\Delta\tau_i,~~~r=1,\dots,k.\\
\tilde{z}^{(k)}(t_i)&=z(t_i)\\
\tilde{z}^{(k)}(\tau_{i(r+1)}) &= \tilde{z}^{(k)}(\tau_{ir}) + A(\tau_{ir},\tilde{z}^{(k)}(\tau_{ir}))\Delta \tau_i + L(\tau_{ir},\tilde{z}^{(k)}(\tau_{ir}))\Delta B_{ir}
\end{alignedat}
\end{align}
where $\Delta B_{ir}=B(\tau_{i(r+1)})-B(\tau_{ir})$. Denote $p_{\theta,t_i,t_{i+1}}^{\tilde{z}^{(k)}}(\cdot|\cdot)$ as the likelihood of $\tilde{z}^{(k)}(t_{i+1})|\tilde{z}^{(k)}(t_{i})$. The following Lemma gives conditions under which the $\tilde z_k(t)$ approximates $z(t)$.
\begin{lemma}
 \cite{Kloeden1992} Under the following conditions:
\begin{itemize}
    \item[(A1)] For all $0<R<\infty,0\le t\le R$, the functions $A(t,\cdot)$ and $L(t,\cdot)$ are Lipschitz continuous in the closed ball $B(\textbf{0},R)$ where $\textbf{0}$ is a vector of 0's.
    \item[(A2)] For all $0<R<\infty$ there exists $0<C_R<\infty$ such that
    \begin{align*}
        \|A(t,x)\|+\|L(t,x)\|\le C_R(1+\|x\|)\quad\forall\ 0\le t\le R,x\in \mathbb{R}^d.
    \end{align*}
    \item[(A3)] $\Sigma(t,x)=L(t,x)L(t,x)^\top$ is positive definite for all $t\ge 0$ and $x\in \mathbb{R}^d$.
\end{itemize}
We have $\tilde{z}^{(k)}(t)\xrightarrow{L_1} z(t)$ for all $t\in [0,T]$ as $k\to \infty$.
\end{lemma}

Note that conditions (A1) and (A2) are satisfied since the components of $A(t,\cdot)$ and $L(t,\cdot)$ are polynomials and square roots of polynomials, respectively. The remaining condition is true as long as the epidemic has not ended, i.e. $i(t)>0$.

Setting $k=1$ in \eqref{3.12}, we have the following scheme
\begin{align}\label{3.14}
\tilde{z}^{(1)}(t_{i+1}) = \tilde{z}^{(1)}(t_i) + A(t_i,\tilde{z}^{(1)}(t_i))\Delta t_i + L(t_i,\tilde{z}^{(1)}(t_i))\Delta B_i
\end{align}
where $\Delta B_i= B(t_{i+1})-B(t_i)$.
With this, we can approximate the likelihood function of $z(t)$ with that of $\tilde{z}^{(1)}(t)$. And due to the construction of $\tilde{z}^{(1)}(t)$ in \eqref{3.14}, we have the following closed form likelihood formula
\begin{align*}
    p_{\theta,t_i,t_{i+1}}^{\tilde{z}^{(1)}}(z_2|z_1)=\phi(z_2-z_1|\Delta t_iA(t_i,z_1), \Delta t_i\Sigma(t_i,z_1))
\end{align*}
where $\phi(\cdot|\mu,\sigma^2)$ is the density of $N(\mu,\sigma^2)$. Now given the data points $x(t_1),\ldots,x(t_M)$, the approximate likelihood is
\begin{align*}
    L_{\tilde{z}^{(1)}}(\theta)=\prod_{i=1}^Mp_{\theta,t_i,t_{i+1}}^{\tilde{z}^{(1)}}(x(t_{i+1})|x(t_i)).
\end{align*}
Another concern here is that in the original SIR model, the states of $x(t)$ are in $[0,1]^2$, which is not the case for $\tilde{z}^{(1)}(t)$. Therefore, in some cases when one or both elements of $x(t_i)$ is 0 or 1, we view it as a censored observation in regard to $\tilde{z}^{(1)}(t)$. These cases are when $s(t_i)=1$ or $j(t_i)=1$ or $s(t_i)=0\wedge s(t_i)\neq0$ or $j(t_i)=0\wedge j(t_i)\neq0$. Now define
\begin{align*}
\begin{alignedat}{3}
\mu(t_i)&=\begin{pmatrix}
    \mu_1(t_i)\\ \mu_2(t_i)
\end{pmatrix}=x(t_{i-1})+A(t_{i-1},x(t_{i-1}))\Delta t_{i-1},\\
\Sigma(t_i)&=\begin{pmatrix}
    \sigma_{11}(t_i) & \sigma_{12}(t_i)\\ \sigma_{12}(t_i) & \sigma_{22}(t_i)
\end{pmatrix}=\Sigma(t_{i-1},x(t_{i-1}))\Delta t_{i-1}
\end{alignedat}
\end{align*}
and
\begin{align}\label{3.18}
\begin{alignedat}{3}
\mu_1^*(t_i)&=\mu_1(t_i)+\dfrac{\sigma_{12}(t_i)}{\sigma_{22}(t_i)}(j(t_{i-1})-\mu_2(t_i)),\\
\mu_2^*(t_i)&=\mu_2(t_i)+\dfrac{\sigma_{12}(t_i)}{\sigma_{11}(t_i)}(s(t_{i-1})-\mu_1(t_i)),\\
(\sigma_1^*(t_i))^2&=\sigma_{11}(t_i)-\dfrac{\sigma_{12}^2(t_i)}{\sigma_{22}(t_i)},\\
(\sigma_2^*(t_i))^2&=\sigma_{22}(t_i)-\dfrac{\sigma_{12}^2(t_i)}{\sigma_{11}(t_i)}.
\end{alignedat}
\end{align}
The terms defined in \eqref{3.18} are just the conditional mean and variance of each component given the other. With this we can write out the likelihood formula for all cases of the data
\begin{align*}
\begin{alignedat}{3}
&p_{\theta,t_{i-1},t_{i}}^{\tilde{z}^{(1)}}(x(t_{i})|x(t_{i-1}))=\\
&\begin{cases}
&\phi(x(t_i)|\mu(t_i),\Sigma(t_i)) \hspace{5.5cm}\text{if }s(t_i),j(t_i)\in(0,1)\\
&\Phi(\boldsymbol{\iota}(x(t_i))|\mu(t_i),\Sigma(t_i))
\hspace{4.9cm}\text{if }s(t_i),j(t_i)\not\in(0,1)\\
&\phi(s(t_i)|\mu_1(t_i),\sigma_{11}(t_i))\Phi(\iota(j(t_i))|\mu_2^*(t_i),(\sigma_2^*(t_i))^2) \quad\text{if }s(t_i)\in(0,1),j(t_i)\not\in(0,1)\\
&\phi(j(t_i)|\mu_2(t_i),\sigma_{22}(t_i))\Phi(\iota(s(t_i))|\mu_1^*(t_i),(\sigma_1^*(t_i))^2) \quad\text{if }s(t_i)\not\in(0,1),j(t_i)\in(0,1)
\end{cases}
\end{alignedat}
\end{align*}
where $\Phi(\boldsymbol{\iota}(x)|\mu,\Sigma)$ and $\Phi(\iota(s)|\mu,\sigma^2)$, respectively, denote the integral of the normal distributions on the corresponding intervals $\boldsymbol{\iota}(x)$ and $\iota(s)$, with
\begin{align*}
    \iota(s) &:= \begin{cases}
        (-\infty,0] \text{ if }s\le 0\\
        [1,\infty)\text{ if }s\ge 1
    \end{cases}, \quad s\in[0,1],\\
    \boldsymbol{\iota}(x) &:= \iota(s)\times\iota(j), \qquad x=(s,j)\in[0,1]^2.
\end{align*}

\subsubsection{Multi-step likelihood approximation}
The Euler-Maruyama approximation described in \eqref{3.14} only makes one jump from one time stamp to the next and the likelihood derived from this is referred to as the 1-step likelihood. Problems with this scheme arise when the time stamps are too far apart or the infection rate changes too quickly between observation times thereby lowering the approximation quality. A solution is to use the k-step scheme in \eqref{3.12} with larger $k$ for better approximation. Note that in the multi-step scheme, we do not know the observations in between the observed times and the likelihood will therefore involve integrating out these values

\begin{theorem}\label{theo3.2}
The likelihood formula for the scheme \eqref{3.12} is as follows
\begin{align}
    p_{\theta,t_i,t_{i+1}}^{\tilde{z}^{(k)}}(z_2|z_1)&=\int \prod_{r=1}^kp_{\theta,\tau_{i(r-1)},\tau_{ir}}^{\tilde{z}^{(1)}}(\xi_{r}|\xi_{r-1})d\xi_1\ldots d\xi_{k-1} \label{3.20}\\
    &= E\left(p_{\theta,\tau_{i(k-1)},t_{i+1}}^{\tilde{z}^{(1)}}(z_2|\tilde{z}^{(k)}(\tau_{i(k-1)}))\left|\tilde{z}^{(k)}(t_i)=z_1\right.\right)\label{3.21}
\end{align}
where $\xi_0=z_1,\xi_k=z_2\in \mathbb{R}^2$.
\end{theorem}
\begin{proof}
Since \eqref{3.20} is by definition, we only need to prove that the right hand side of \eqref{3.20} equals to \eqref{3.21}
\begin{align*}
\int \prod_{r=1}^k & p_{\theta,\tau_{i(r-1)},\tau_{ir}}^{\tilde{z}^{(1)}}(\xi_{r}|\xi_{r-1})d\xi_1\ldots d\xi_{k-1}\\ 
&\ \ =\int p_{\theta,\tau_{i(k-1)},\tau_{ik}}^{\tilde{z}^{(1)}}(\xi_{k}|\xi_{k-1})\prod_{r=1}^{k-1}p_{\theta,\tau_{i(r-1)},\tau_{ir}}^{\tilde{z}^{(1)}}(\xi_{r}|\xi_{r-1})d\xi_1\ldots d\xi_{k-1}\\
&\stackrel{Fubini}{=}\int p_{\theta,\tau_{i(k-1)},\tau_{ik}}^{\tilde{z}^{(1)}}(\xi_{k}|\xi_{k-1})\left(\int\prod_{r=1}^{k-1}p_{\theta,\tau_{i(r-1)},\tau_{ir}}^{\tilde{z}^{(1)}}(\xi_{r}|\xi_{r-1})d\xi_1\ldots d\xi_{k-2}\right) d\xi_{k-1}\\
&\ \ =\int p_{\theta,\tau_{i(k-1)},\tau_{ik}}^{\tilde{z}^{(1)}}(\xi_{k}|\xi_{k-1})p_{\theta,\tau_{i0},\tau_{i(k-1)}}^{\tilde{z}^{(k-1)}}(\xi_{k-1}|\xi_0)d\xi_{k-1}\\
&\ \ =\int p_{\theta,\tau_{i(k-1)},t_{i+1}}^{\tilde{z}^{(1)}}(z_2|\xi_{k-1})p_{\theta,t_i,\tau_{i(k-1)}}^{\tilde{z}^{(k-1)}}(\xi_{k-1}|z_1)d\xi_{k-1}\\
&\ \ =E\left(p_{\theta,\tau_{i(k-1)},t_{i+1}}^{\tilde{z}^{(1)}}(z_2|\tilde{z}^{(k)}(\tau_{i(k-1)}))\left|\tilde{z}^{(k)}(t_i)=z_1\right.\right).
\end{align*}
\end{proof}
Using the law of large numbers and the expression \eqref{3.21} in theorem \ref{theo3.2}, we have the following procedure to approximate the multi-step likelihood
\begin{itemize}
    \item Simulate $B$ sample paths using \eqref{3.12} to get $\tilde{z}^{(k)}_1(\tau_{i(k-1)}),\ldots,\tilde{z}^{(k)}_B(\tau_{i(k-1)})$.
    \item By law of large numbers, we have
    \begin{align*}
        \dfrac{1}{B}\sum_{b=1}^Bp_{\theta,\tau_{i(k-1)},t_{i+1}}^{\tilde{z}^{(1)}}(z_2|\tilde{z}^{(k)}_b(\tau_{i(k-1)})) \xrightarrow{a.s} E\left(p_{\theta,\tau_{i(k-1)},t_{i+1}}^{\tilde{z}^{(1)}}(z_2|\tilde{z}^{(k)}(\tau_{i(k-1)}))\left|\tilde{z}^{(k)}(t_i)=z_1\right.\right).
    \end{align*}
\end{itemize}
We can parallelize the generation of these paths by using multiple processors, one for each path. The trade-off for using the multi-step likelihood is the increased computational time due to the simulations.

\subsection{Tau leaping approximation}
For this method, we go back to the definition of the stochastic SIR model
\begin{align}\label{3.23}
    &X(t) = X(0) + \begin{pmatrix}-1\\ 1 \end{pmatrix} Pois_1\left(\int_0^t\beta(s)\dfrac{S(s)I(s)}{N}ds\right) + \begin{pmatrix}0\\ -1 \end{pmatrix} Pois_2\left(\int_0^t\gamma I(s)ds\right)
\end{align}
where $Pois_1,Pois_2$ are independent standard Poisson processes. Tau leaping is a method for approximating \eqref{3.23} with a process $\tilde{X}^{(k)}$ defined by a scheme similar to the Euler-Maruyama method 
\begin{align}
\begin{alignedat}{3}
\tau_{ir}&=t_i+r\dfrac{\Delta t_i}{k}=t_i+r\Delta\tau_i,\\
\tilde{X}^{(k)}(t_i)&=X(t_i),\\
\tilde{X}^{(k)}(\tau_{i(r+1)}) &= \tilde{X}^{(k)}(\tau_{ir}) + \begin{pmatrix}-1\\ 1 \end{pmatrix} Pois_1\left(\Delta\tau_i\beta(\tau_{ir})\dfrac{S(\tau_{ir})I(\tau_{ir})}{N}\right)\\
&\qquad\qquad\qquad + \begin{pmatrix}0\\ -1 \end{pmatrix} Pois_2\left(\Delta\tau_i\gamma I(\tau_{ir})\right).
\end{alignedat}
\end{align}
With this the likelihood function can be approximated using the transition probabilities of $\tilde{X}^{(k)}$. Specifically, for $k=1$ we have
\begin{align}
\begin{alignedat}{3}
    L_{\tilde{X}^{(1)}}(\theta)&= \prod_{i=1}^M P_{\theta,t_i,t_{i+1}}^{\tilde{X}^{(1)}}(X(t_{i+1})|X(t_i))\\
    &=\prod_{i=1}^M f\left(\Delta W_i\left|\Delta t_i\beta(t_i)\dfrac{S(t_i)I(t_i)}{N}\right.\right)f(\Delta Y_i|\Delta t_i\gamma I(t_i))
\end{alignedat}
\end{align}
where $P_{\theta,t_i,t_{i+1}}^{\tilde{X}^{(1)}}(\cdot|\cdot)$ denotes the probability mass function of $\tilde{X}^{(1)}(t_{i+1})|\tilde{X}^{(1)}(t_i)$, $\Delta W_i=S(t_i)-S(t_{i+1})$, $\Delta Y_i = S(t_i)-S(t_{i+1})+I(t_i)-I(t_{i+1})$ and $f(\cdot|\lambda)$ is the probability mass function of a Poisson distribution with rate $\lambda$.

To compute the multi-step likelihood approximation, we use the same procedure and parallelization devised for diffusion processes with the path simulation method and one step likelihood formula changed to that of tau leaping. However, since the Poisson variables generated each step have different rates across the simulated paths and the implementation of our method in \verb|R| only uses one processor, we cannot vectorize the generation of these paths like in the diffusion approximation case, where paths can be updated by generating a vector of standard normal variables and updating the paths using \eqref{3.12}. Therefore, we do not compute results using multi-step Tau leaping in our implementation due to it being too time-consuming.

\section{Regression Spline (RS) Framework}\label{RS}
\subsection{Model construction}
Consider the stochastic SIR model, as defined in section \ref{background}, with infection rate function $\beta(t)$ and constant recovery rate $\gamma$. Our goal is to estimate both $\beta(t)$ and $\gamma$ using discretely observed data of the number of susceptible and infected individuals. To this end, a B-spline basis is used for modeling $\beta(t)$. In summary, the model can be written as follows
\begin{align*}
&X(t) = (S(t),I(t)):\text{ stochastic SIR model with rates } \beta(t),\gamma\\
&X(t_1),X(t_2),\ldots,X(t_M): \text{ observed states at times }t_1,t_2,\ldots,t_M\\
&\gamma = \theta_1,\quad 
\beta(t) = \sum_{i=1}^{K+d+1}\theta_{i+1}\phi_{i,d}(t)
\end{align*}
where $K,d$ are the number of knots and degree of the B-spline basis, respectively, and $\theta_i$ are the coefficients.
\subsection{Parameter Estimation}
With a method to approximately compute the likelihood function, the maximum likelihood estimate (MLE) for the model parameters can be found using built-in \verb|R| functions such as \verb|optim|. Before that, we need to fine tune the hyperparameters, specifically the number of knots and their locations. For knot location, we can use the knot placement method in \cite{Yeh2020}. In the paper, we are given the values of the curve $\beta(t)$ at times $u_0,\ldots,u_m$ and the goal is to find the knots $\kappa_1,\ldots,\kappa_K$ for the degree $d$ B-spline basis used to estimate $\beta(t)$. The method determines knots so that the cumulative rate of change in $\beta(t)$, where that rate of change is measured by $\beta^{(d)}(t)$, is constant between knots. 
This is achieved by the following steps:
\begin{enumerate}
    \item[1.] Calculate the $(d+1)^{th}$ derivative $\beta^{(d)}(t)$ of $\beta(t)$ using the formula
    \begin{align}\label{2.24}
        \beta^{(j+1)}(u^{(j+1)}_i)=\dfrac{\beta^{(j)}(u^{(j)}_{i+1})-\beta^{(j)}(u^{(j)}_{i})}{u^{(j)}_{i+1}-u^{(j)}_{i}},\quad u^{(j+1)}_{i}=\dfrac{1}{2}(u^{(j)}_{i}+u^{(j)}_{i+1})
    \end{align}
    where $\beta^{(0)}(t)=\beta(t)$. Note that \eqref{2.24} implies that each derivative level has its own time stamps which are the midpoints of the previous level's time stamps.
    \item[2.] Calculate the feature function $f(u)$. The feature function $f(u)$ is defined as the piecewise linear function that satisfies
    \begin{align*}
        f_i=f(\bar{u}_i)=\begin{cases}
            0 &\text{if }i=0,m-d\\
            \|\beta^{(d+1)}(u^{(d+1)}_i)\|^{1/(d+1)} &\text{otherwise}
        \end{cases}
    \end{align*}
    where $\bar{u}_0=u_1,\bar{u}_{m-d}=u_m$ and $\bar{u}_i=u^{(d+1)}_i$ for $0<i<m-d$. \item[3.] Calculate the cumulative feature function $F(u)$. $F(u)$ is defined as the integral of $f(u)$, i.e.
    \begin{align*}
        F(u)=\int\limits_{-\infty}^uf(v)dv.
    \end{align*}
    \item[4.] Obtain knot locations from the feature curve by setting $\kappa_j=F^{-1}(j\Delta F)$ where $\Delta F=\max\limits_u \frac{F(u)}{k-1}$. In other words, divide the feature curve into segments with equal amount of increase and set the corresponding time stamps as knots. Computation of $F^{-1}(u)$ can be simplified by pretending $F(u)$ is a piecewise linear function and values at $\bar{u}_i$ calculated using trapezoid rule for $f_i$.
\end{enumerate}
The last four plots in Figure \ref{knotplacement} show how these steps are carried out given the data in the first plot.

\begin{figure}
    \centering
    \begin{tikzpicture}[
  node distance=1.5cm,
  >={Triangle[angle=60:1pt 2]},
  shorten >= 2pt,
  shorten <= 2pt,
  arrow/.style={
    ->,
    arrowblue,
    line width=5pt
  }
]
\ImageNode[label={-90:Input data}]{A}{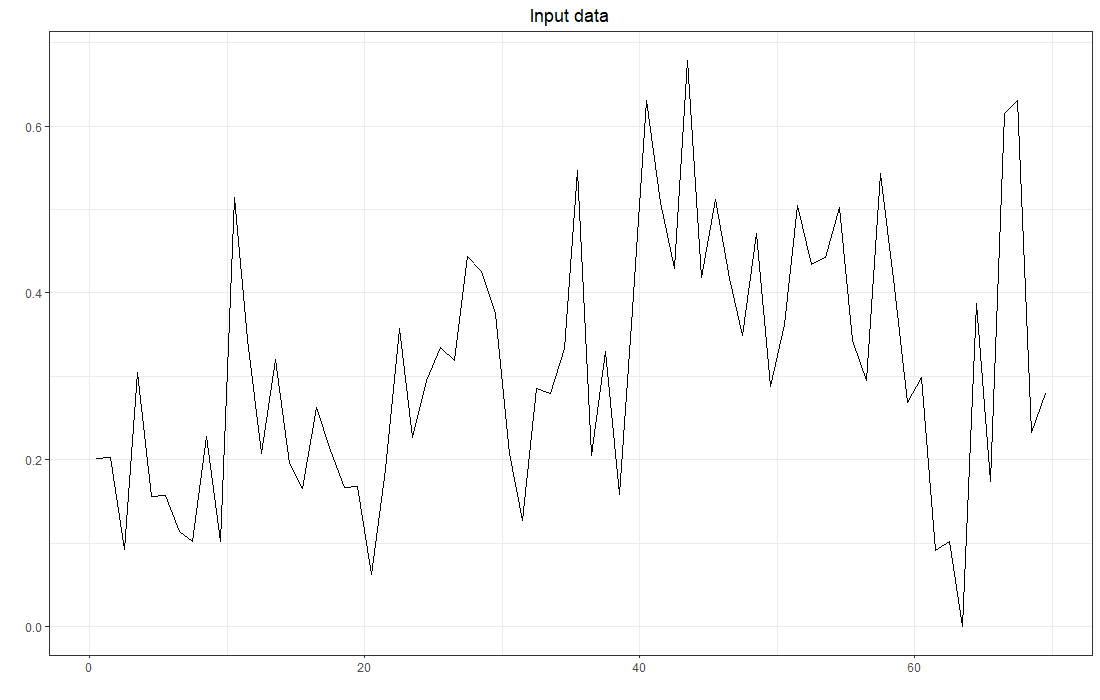}
\ImageNode[label={-90:Forth Derivative},right=of A]{B}{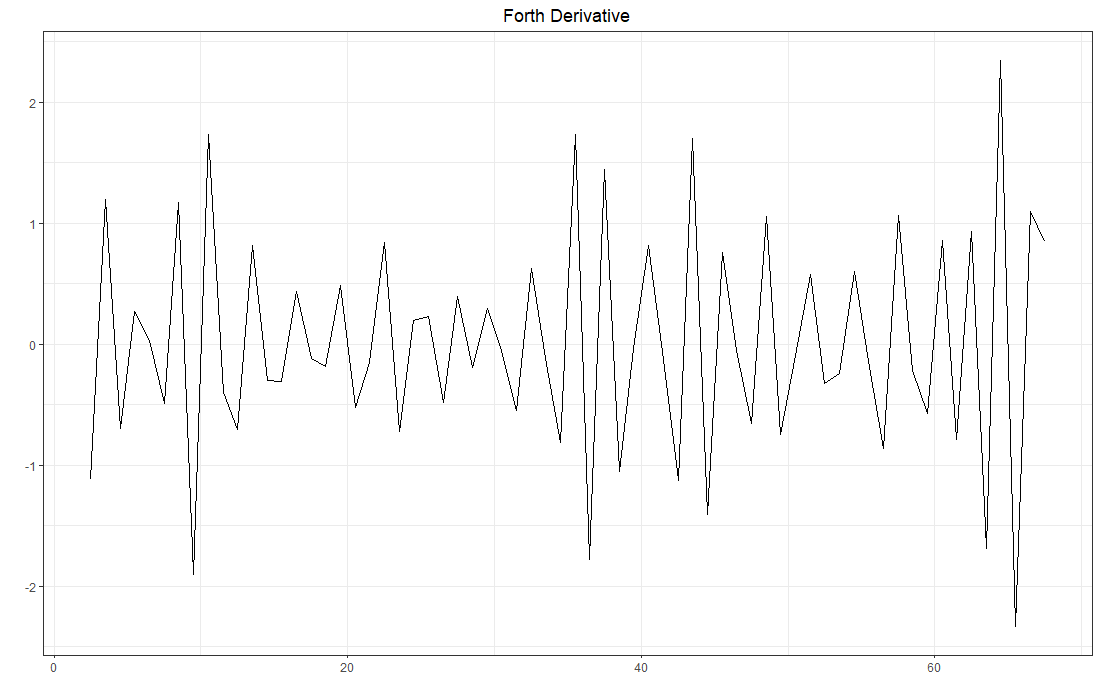}
\ImageNode[label={-90:Feature Function},right=of B]{C}{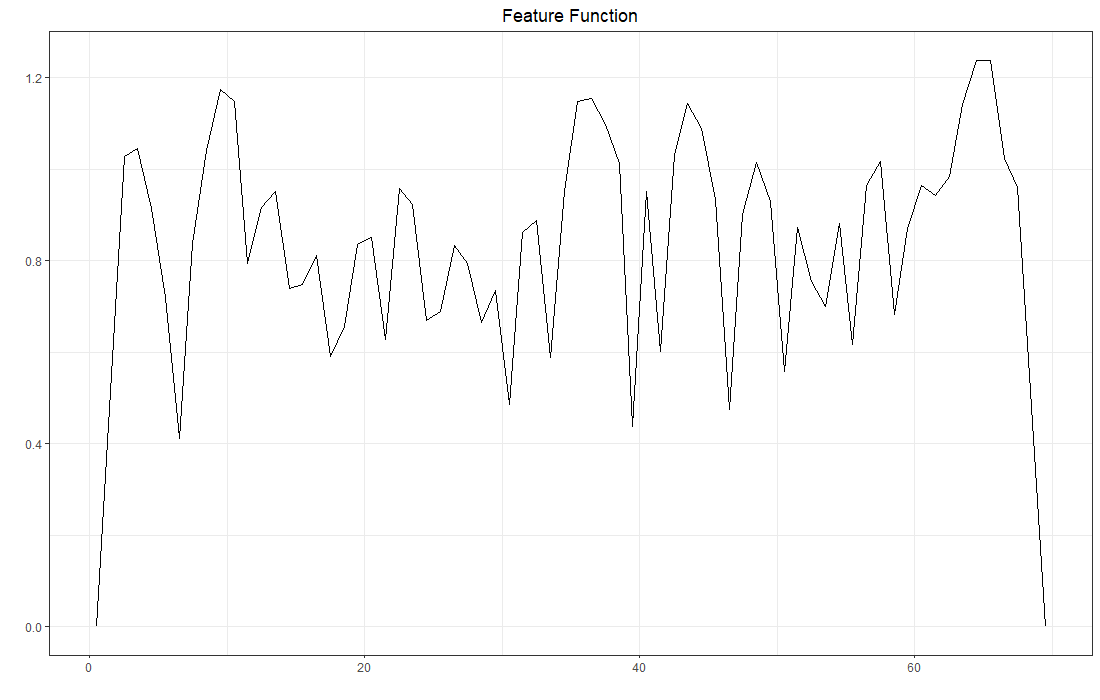}
\ImageNode[label={-90:Cumulative Feature Function},below left=of C]{D}{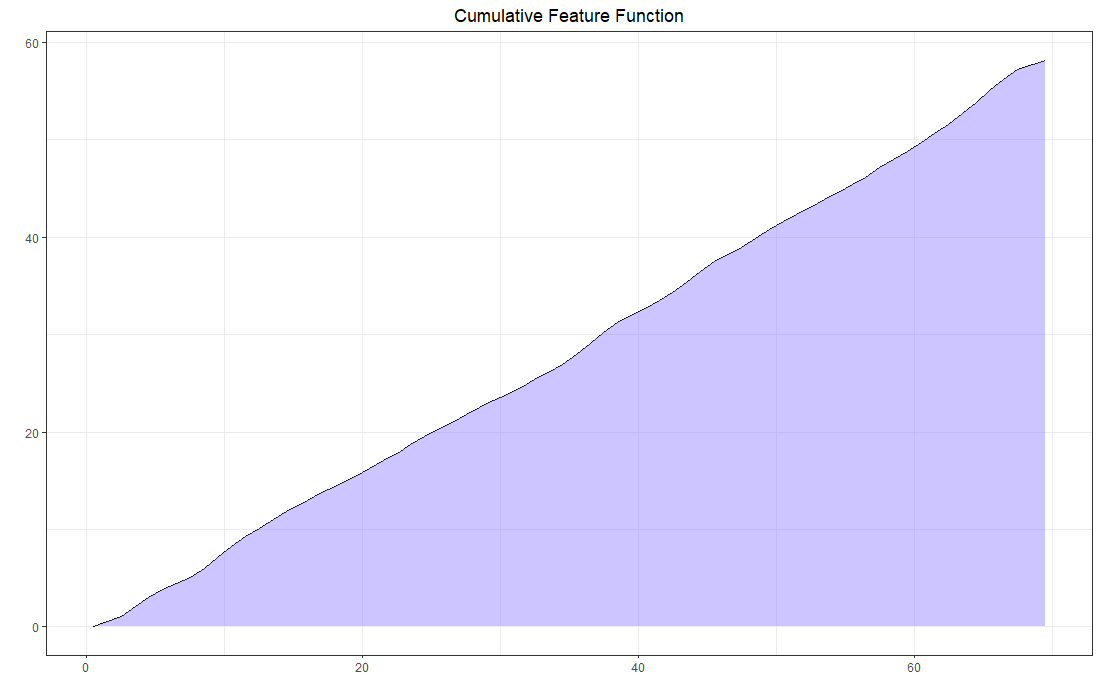}
\ImageNode[label={-90:Knot Placement},left=of D]{E}{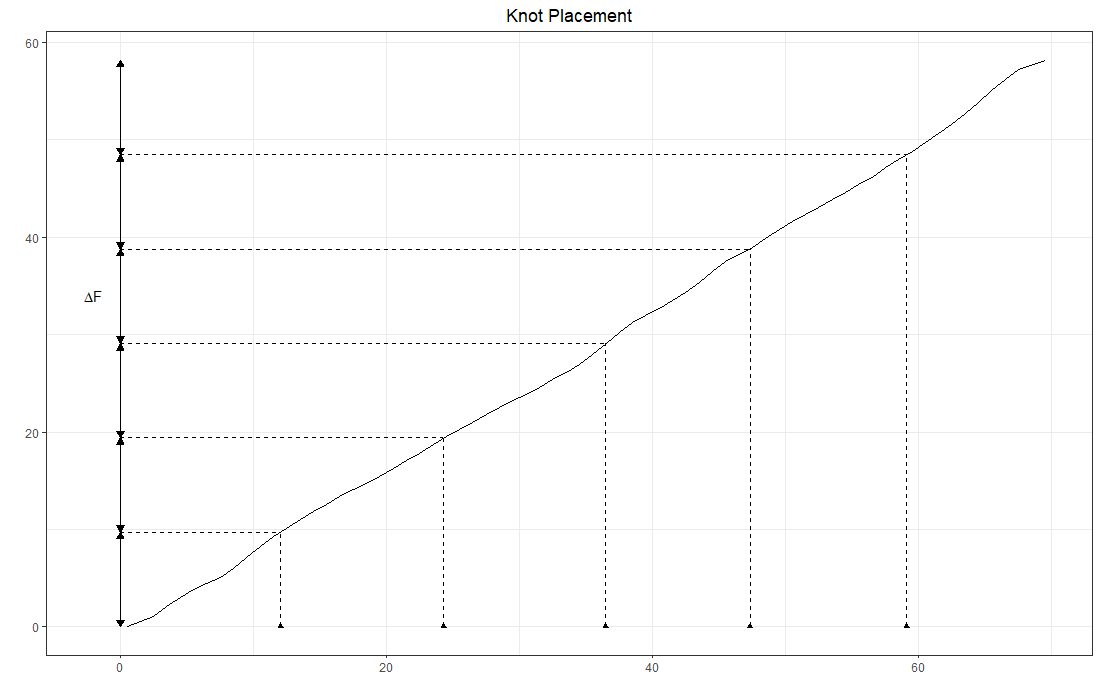}

\draw[arrow]
  (A) -- (B);
\draw[arrow]
  (B) -- (C);
\draw[arrow]
  (C) -- (D);
\draw[arrow]
  (D) -- (E);
\end{tikzpicture}
    \caption{Step by step illustration of knot placement method when using a degree 3 B-spline basis.}
    \label{knotplacement}
\end{figure}

For the number of knots to use, we choose the one with the best Bayesian information criterion (BIC) through a forward selection scheme. The idea is to increase the number of knots one at a time, and stop when none of the last 3 models (numbers of knots) has improved on the best value of the criterion.

\subsubsection{Moving average rate estimate}
Note that in order to use the knot placement method, we need a time series input of the target function, which we do not have for $\beta(t)$. To resolve this, "true" values of $\beta(t)$ are created by estimating the moving average rates between observations. Consider $w$ consecutive observations $x(t_i),\ldots,x(t_{i+w})$, we now build a mini model by assuming that the infection rate is a constant in this period. Then the estimate for $\beta$ using these data points will be the guess for the true value of $\beta((t_i+t_{i+w})/2)$. Next, the procedure is repeated over all windows of $w$ consecutive observations to get the curve values for knot selection.\\
For example, if the whole data is $x(0),x(1),\ldots,x(30)$ and window size is $w=3$, then $x(0),x(1),x(2)$ are used to estimate the value of $\beta(1)$; $x(1),x(2),x(3)$ for $\beta(2)$ and so on. The idea for this procedure is pretty similar to a moving average of a time series but the average series is for the hidden infection rate.

\section{Confidence interval}\label{CI}
A parametric bootstrap scheme is used to find the point wise confidence interval for $\beta(t)$. A complication is that the duration, $T_b$, of an epidemic for a bootstrap simulated data set can be shorter than the duration for the original data. Such bootstrap samples provide no information about $\beta(t)$ for any $t>T_b.$ This can lead to uninformative or even bad intervals for the infection rate function. An example is when both $I(t_1)$ and the estimated rate $\hat{\beta}(t)$ is small in the beginning, which can happen when the spline degree is 2 or higher, leading to many bootstrap samples terminating too early. A solution is to discard simulated paths that terminated early and use the ones that survived until the final observed time $t_M$ as bootstrap samples.

\subsection{Interval smoothing}
There are reasons to expect that utilizing information from neighbouring time points will improve confidence bound performance. If $\beta(t)$ is continuous at $t$, then the bounds at neighbouring time points should provide additional information and hence a way of reducing variability of confidence bounds. If, on the other hand, $t$ is a discontinuity point or a point of rapid increase, estimation uncertainty might lead to that jump being estimated to the right or the left of t. By incorporating information from  adjacent points that have the jump, one recognizes that $\beta(t)$ might plausibly have been much smaller or higher than suggested by its bounds. In interval smoothing we smooth out the pointwise confidence interval using adjacent time stamps. In particular, we consider three different ways of smoothing: weighted smoothing, sample smoothing and min-max smoothing.

The first way is to use the interval values themselves. Let $L_{t_i},U_{t_i}$ be the lower and upper bounds of the confidence interval for $\beta(t_i)$. Then the smoothed confidence interval $[\bar{L}_{t_i},\bar{U}_{t_i}]$ is calculated as the weighted sum of adjacent bounds as follows
\begin{align}
    \bar{L}_{t_i} = \dfrac{\sum_j w(t_i,t_j)L_{t_j}}{\sum_j w(t_i,t_j)},\quad
    \bar{U}_{t_i} = \dfrac{\sum_j w(t_i,t_j)U_{t_j}}{\sum_j w(t_i,t_j)}
\end{align}
where the weighting function $w(\cdot,\cdot)$ is the normal kernel
\begin{align*}
    w(x,y) = \phi(x-y).
\end{align*}

The second way is to combine the $\beta$ values from the bootstrap samples at adjacent time points and use them as the the samples for the middle point. Specifically, we use $\hat{\beta}^*_1(t_{i-1}),\ldots,\ \hat{\beta}^*_B(t_{i-1}),\ \hat{\beta}^*_1(t_{i}),\ldots,$ $ \hat{\beta}^*_B(t_{i}),\ \hat{\beta}^*_1(t_{i+1}),\ldots,\ \hat{\beta}^*_B(t_{i+1})$ as samples to construct the confidence interval of $\beta(t_i)$ instead of just $\hat{\beta}^*_1(t_{i}),\ldots,\hat{\beta}^*_B(t_{i})$. Here $\hat{\beta}^*_b(t)$ denotes the estimated infection rate at $t$ for the $b^{th}$ bootstrap sample and $B$ is the number of bootstrap samples. The intuition behind this step is to improve the coverage rate at places where there are significant changes in the infection rate.

The third method is to simply widen the bounds by setting the new upper bounds as the largest of all the surrounding bounds and the new lower bounds as the smallest of all the surrounding bounds. Specifically, the new interval $[\bar{L}_{t_i},\bar{U}_{t_i}]$ for $\beta(t_i)$ is
\begin{align}
    \bar{L}_{t_i} = \min\{L_{t_{i-1}},L_{t_{i}},L_{t_{i+1}}\},\quad
    \bar{U}_{t_i} = \max\{U_{t_{i-1}},U_{t_{i}},U_{t_{i+1}}\}.
\end{align}

\section{Simulation Study}\label{sim_study}
In this section, various performance aspects of the proposed model are investigated using simulated data. The data sets are generated using the \verb|R| package \verb|ssar| \cite{ssar}, which employs the Gillespie algorithm for exact simulation of the stochastic SIR model. Specifically, we mainly look at 4 typical epidemic patterns where the infection rate is constant (Simulation 1), increasing, decreasing, going up then down. We also run a fifth simulation where the infection rate is smoothly increasing. Each data set consists of 71 data points, the recovery rate is set to 0.1 for all 5 simulations, the infection rate for the constant case is set to 0.3 and each simulation is repeated 100 times. The infection rates for simulations 2 to 5 are plotted in Figure \ref{Figsim}. In addition, the populations are all set to $N=10000$ with initial proportion of susceptible and infected at $99\%$ and $1\%$, respectively.

\begin{figure}[h]
     \centering
     \begin{subfigure}[b]{0.48\textwidth}
         \centering
         \includegraphics[width=\textwidth]{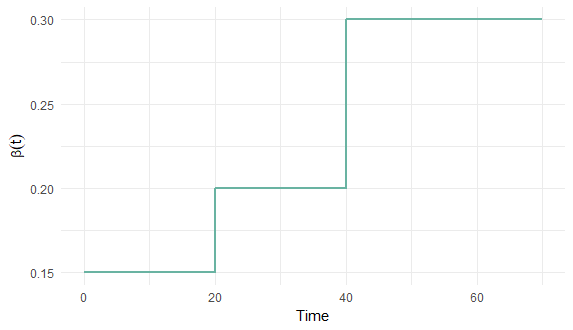}
         \caption{Simulation 2}
         \label{sim2}
     \end{subfigure}
     \hfill
     \begin{subfigure}[b]{0.48\textwidth}
         \centering
         \includegraphics[width=\textwidth]{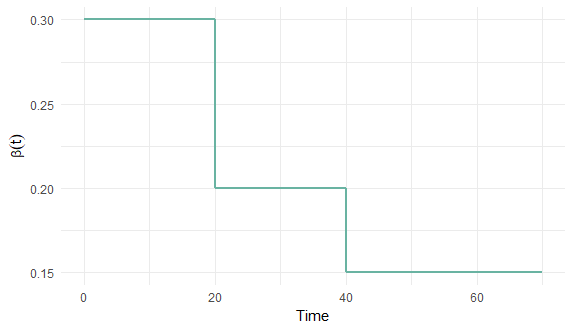}
         \caption{Simulation 3}
         \label{sim3}
     \end{subfigure}
     \\
     \centering
     \begin{subfigure}[b]{0.48\textwidth}
         \centering
         \includegraphics[width=\textwidth]{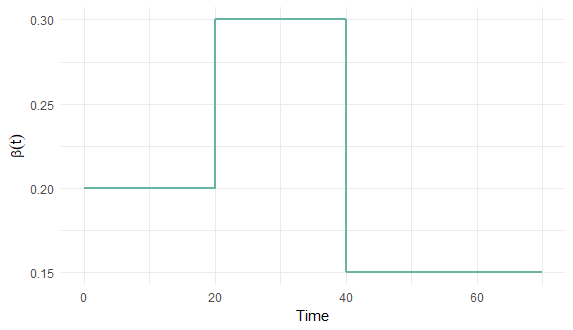}
         \caption{Simulation 4}
         \label{sim4}
     \end{subfigure}
     \hfill
     \begin{subfigure}[b]{0.48\textwidth}
         \centering
         \includegraphics[width=\textwidth]{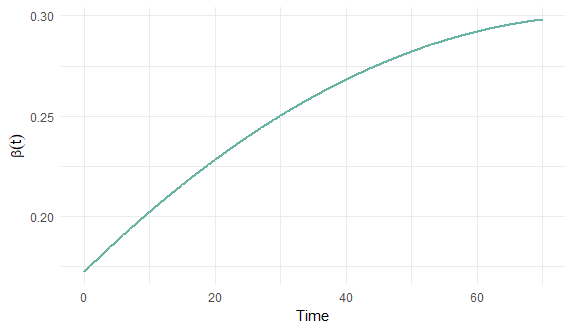}
         \caption{Simulation 5}
         \label{sim5}
     \end{subfigure}
        \caption{Infection rate functions of simulations 2 to 5.}
        \label{Figsim}
\end{figure}

\subsection{Likelihood choice}
In this subsection, we look at the performance of the two methods for likelihood approximation: Diffusion approximation and Tau leaping. To this end, we use the integrated mean squared error (IMSE) between the estimated $\hat{\beta}(t)$ and true infection rate $\beta(t)$ to measure the performance quality of each method, i.e.
\begin{align}
    IMSE(\hat{\beta},\beta) = \int(\hat{\beta}(t)-\beta(t))^2dt.
\end{align}
Figure \ref{figlikelihood} shows the estimation quality comparison between the two single-step likelihood approximation methods when the Regression Spline framework is used across different disease patterns. Based on that, the Tau leaping method performs slightly better than diffusion approximation in all settings. This can be attributed to the former only having one approximation layer (true process is approximated by Euler-Maruyama scheme) while the latter has two approximation layers (true process is approximated by a diffusion process then diffusion process is approximated by Euler-Maruyama scheme). Therefore, we will focus on the Tau leaping likelihood approximation in results below for single-step likelihood. 

\begin{figure}[h]
    \centering
    \includegraphics[width=\textwidth]{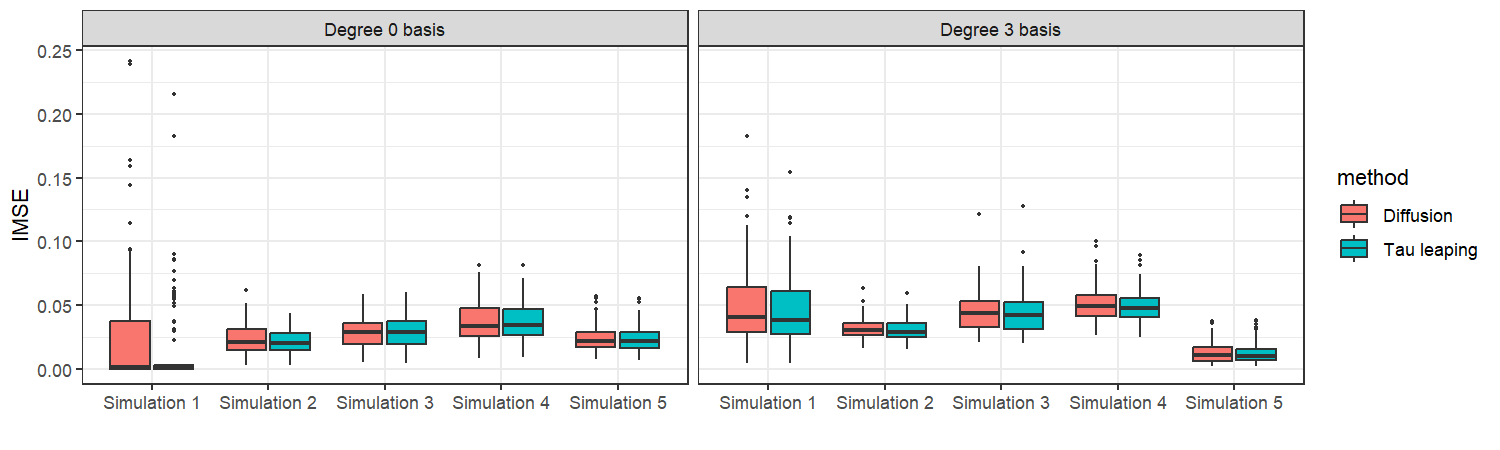}
    \caption{IMSE for different likelihood approximation methods using Regression Spline framework with different B-spline bases.}
    \label{figlikelihood}
\end{figure}

\subsection{Estimations}
To see how each method performs at different stages of an epidemic we plotted the estimations in time in Figures \ref{FigSim_res_0} and \ref{FigSim_res_3}. The solid lines are the true rates, the dashed lines and two bands are the average, $5\%$ and $95\%$ quantiles of the estimations, respectively. 
Here we consider three methods: 2 using 1-step tau-leaping likelihood approximation, one for each framework, and 1 using 20-step diffusion likelihood approximation with 100 sample paths for the RS framework. The 20-step scheme seems to outperform the other methods in terms of bias since this approximation allows the model to better capture the changes in the compartments between 2 consecutive time points. The trade-off is that multi-step methods have larger variance due to accumulating errors from the simulated paths. Looking at estimates and width of the intervals giving the 5th and 95th percentiles, we see that methods sometimes under-estimate when the first change point occurs in Simulations 1 and 4 and sometimes over-estimate the timing of the second change in Simulation 3. These are always associated with the smaller jump of the two but it is not clear what the reason is for the direction of the uncertainty.  % add comparison between RS and PG

Comparing to the results in Figure \ref{FigSim_res_0}, which use a degree 0 (step-wise) basis, to the ones in Figure \ref{FigSim_res_3}, which use degree 3 (smooth) basis, we can see that the step-wise basis works better for the step-wise truths and the smooth basis for the smooth truth. This means that there is a model specification aspect to consider when choosing the method. Perhaps surprisingly, however, the degree 0 splines do fairly well at estimating a continuously increasing curve. Similarly, although the degree 3 splines do not as accurately indicate the rapid rates of increase in the change point simulations 1-3, they do get the gross features correct. The one-step estimation shows a bias with a constant term. 

\begin{figure}[h]
    \centering
    \includegraphics[width=\textwidth]{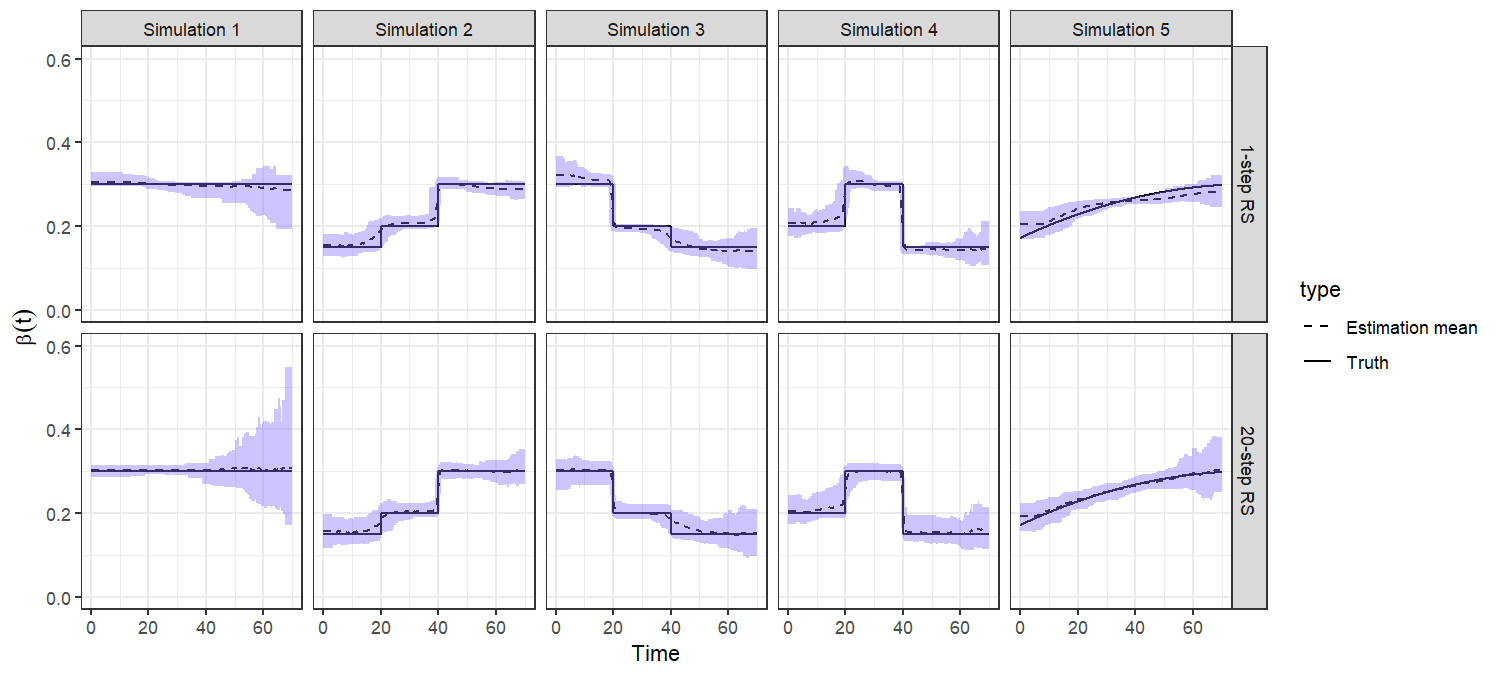}
    \caption{Estimation performance in time for different methods using degree 0 B-spline basis.}
    \label{FigSim_res_0}
\end{figure}

\begin{figure}[h]
    \centering
    \includegraphics[width=\textwidth]{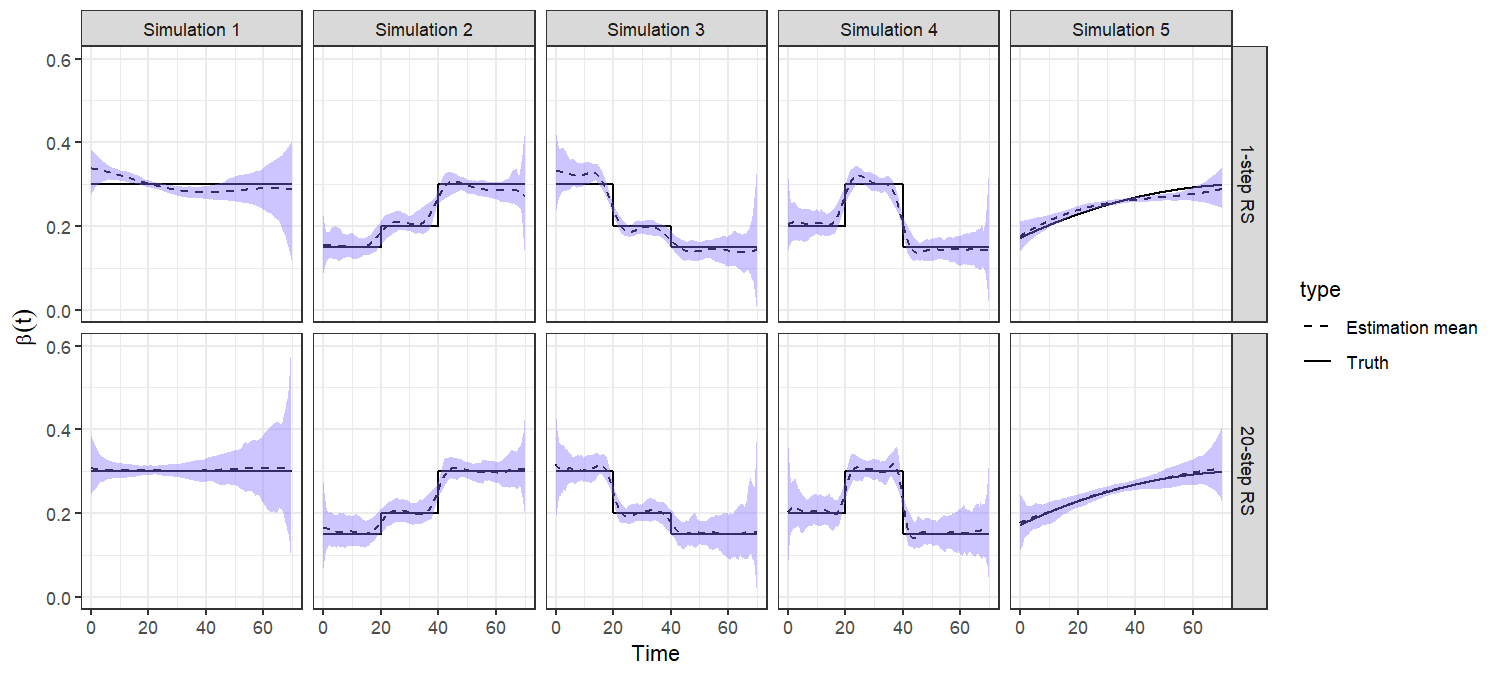}
    \caption{Estimation performance in time for different methods using degree 3 B-spline basis.}
    \label{FigSim_res_3}
\end{figure}

\subsection{Confidence interval}
Figures \ref{Fig_CIbias0} and \ref{Fig_CIbias3} show the coverage rates coverage rates (proportions of times the intervals contained the true rate) of 95\% pointwise confidence intervals for $\beta(t)$, with and without bias correction. The biases and small variation of estimation evident in Figures 5 and 6 (Single step approximation) are reflected in the parametric bootstraps used to construct the intervals and lead to bad under-coverage in time periods, where those time periods largely correspond to regions of bias combined with small variation in Figures 5 and 6. For Simulations 1-3, not surprisingly, under-coverage is particularly bad near change points but we see that this extends to regions in between change points and to regions of bias/low-variation where there are no change points. 

For both kinds of intervals, bias correction does appear to generally help with coverage rates, especially with a degree 3 basis. Even when the coverage rates of some time points suffer, the overall coverage across all time points still increase. Therefore, we suggest using bias correction for both percentile and normal confidence intervals.

\begin{figure}[h]
    \centering
    \includegraphics[width=\textwidth]{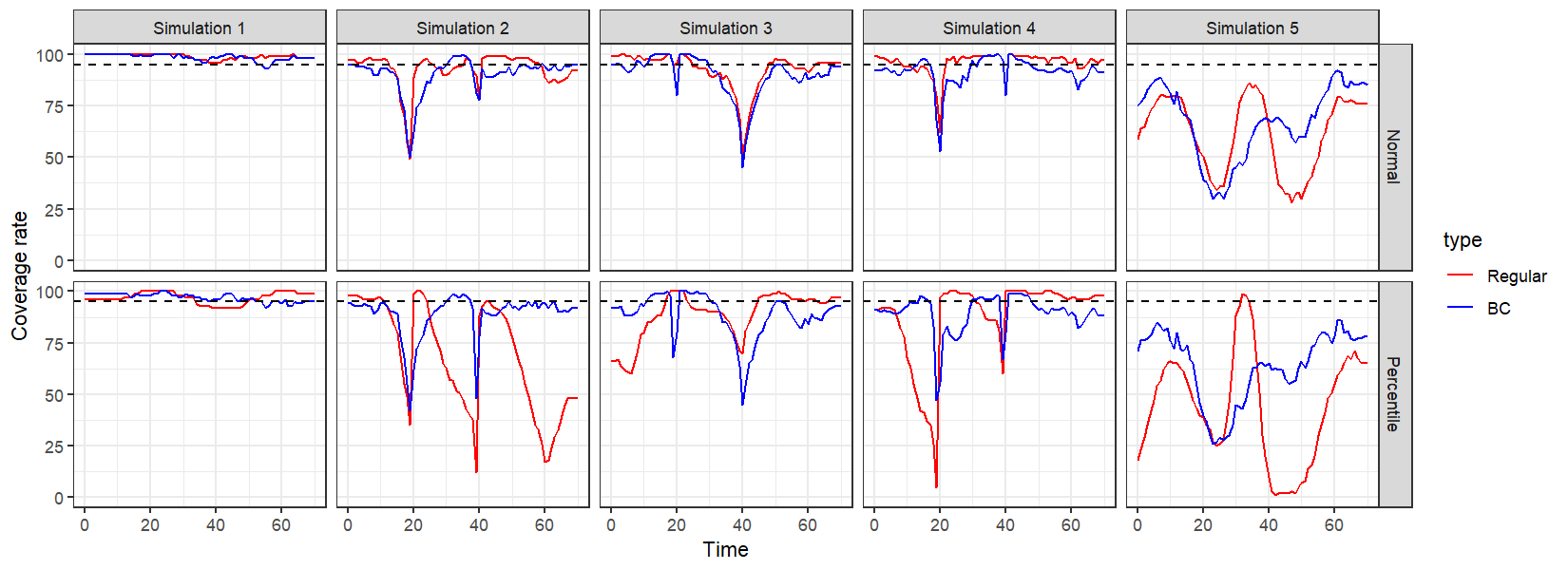}
    \caption{Coverage rate comparison of confidence intervals using degree 0 B-spline basis with and without bias correction (BC).}
    \label{Fig_CIbias0}
\end{figure}

\begin{figure}[h!]
    \centering
    \includegraphics[width=\textwidth]{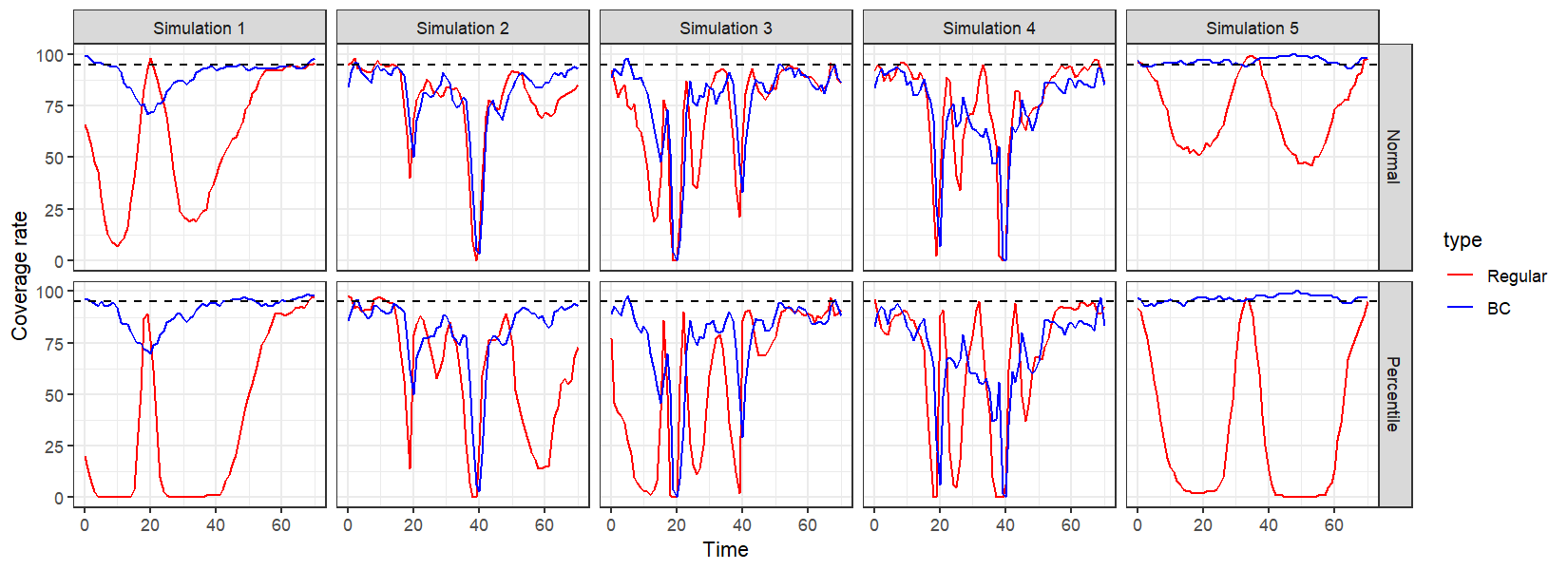}
    \caption{Coverage rate comparison of confidence intervals using degree 3 B-spline basis with and without bias correction (BC).}
    \label{Fig_CIbias3}
\end{figure}

Another aspect to look at is how bias correction helps when used in conjunction with the interval smoothing methods. Figures \ref{Fig_CIboth0} and \ref{Fig_CIboth3} illustrate the performance of the interval smoothing methods with and without bias correction along with the intervals where only bias correction is applied. These show that when both adjustments are applied, the coverage rates tend to be better than when only one is applied. In addition, the min-max smoothing method has the best coverage
rates out of the three in most cases. This is to be expected as min-max smoothing widens the intervals, guaranteeing improvement in coverage rates in cases where they are too low. The performance of the other two is interesting as weighted smoothing works better when a degree 0 basis is used whereas sample smoothing works better for a degree 3 basis. The reason may lie in the nature of each basis. A degree 0 basis gives step-wise constant estimates so weighted smoothing can improve the smoothness between intervals at different time points. A degree 3 basis, on the other hand, has smoothness but lacks the ability to rapidly change its values like a degree 0 basis, which makes sample smoothing more useful since it helps expand the bootstrap sample range in places where the infection rate changes quickly.

\begin{figure}[h]
    \centering
    \includegraphics[width=\textwidth]{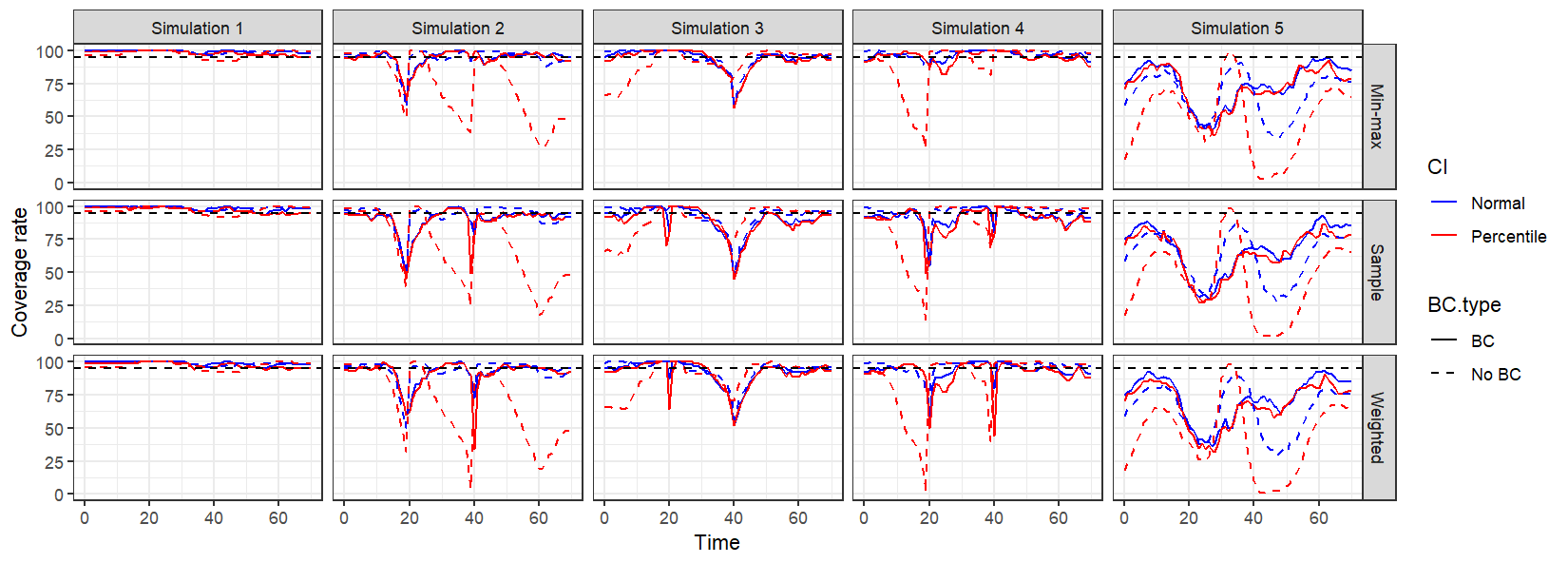}
    \caption{Coverage rate comparison of smoothed confidence intervals using degree 0 B-spline basis with and without bias correction (BC).}
    \label{Fig_CIboth0}
\end{figure}

\begin{figure}[h]
    \centering
    \includegraphics[width=\textwidth]{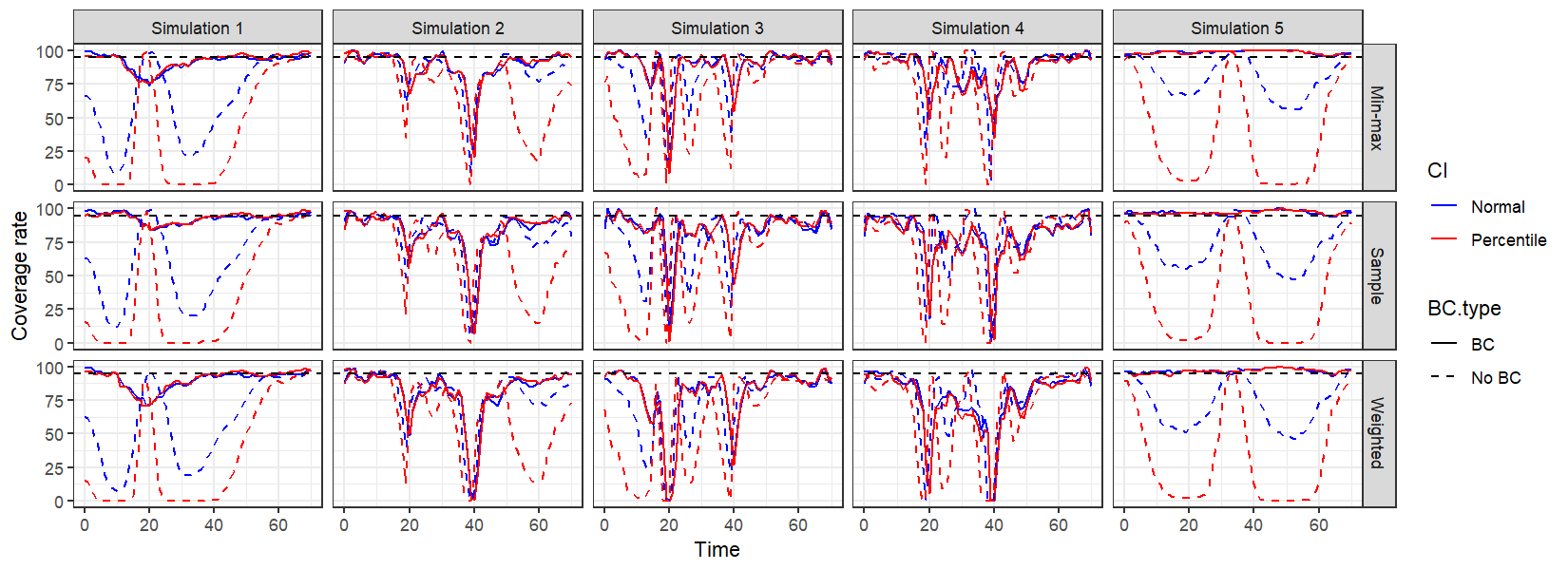}
    \caption{Coverage rate comparison of smoothed confidence intervals using degree 3 B-spline basis with and without bias correction (BC).}
    \label{Fig_CIboth3}
\end{figure}

To summarize, the optimal combination uses bias correction together with min-max smoothing. Figure \ref{Fig_CI_final} shows the coverage rate of the confidence interval with this combination.

\begin{figure}[h]
    \centering
    \includegraphics[width=\textwidth]{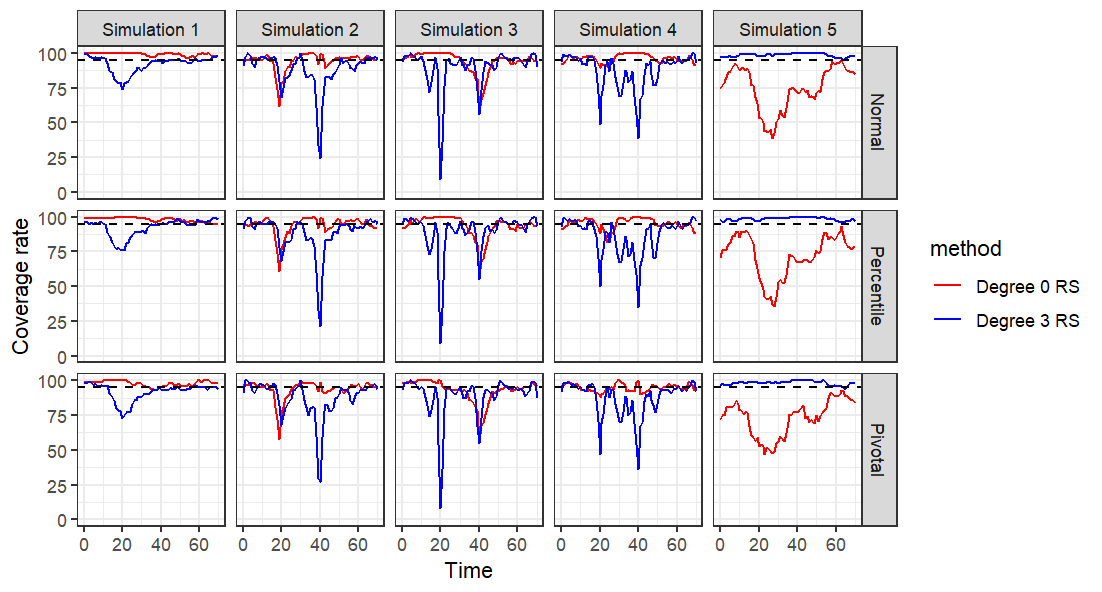}
    \caption{Coverage rate of the optimal combination (bias correction and min-max smoothing) for confidence interval.}
    \label{Fig_CI_final}
\end{figure}

\section{Application}\label{application}
In this section, we will be estimating the basic reproduction number $R_0(t)=\beta(t)/\gamma$ of the COVID-19 data from Ontario between January 2020 and January 2022. The goal is to see how our proposed framework performs for a period in which multiple waves have occurred.

The data was obtained from \cite{Covid_data} in early 2022 when the number of active cases was still recorded. The population is set to $N=14,223,942$, which is the population of Ontario in 2021 according to \cite{Population_data}. For our model, the number of susceptible $S$ is obtained by subtracting the cumulative cases from the population and the number of infected $I$ is the number of active cases in the data.
We use the 1-step diffusion and tau leaping method for likelihood approximation, BIC for model selection, window sizes 2 (daily),4 (3 days) and 8 (weekly), and both degree 0 and 3 B-spline bases. The estimates are plotted in Figures \ref{FigCovidRes0} and \ref{FigCovidRes3}. For a degree 0 basis, the results from the daily and weekly window are more simple with fewer change points. For a degree 3 basis, estimates agree across all window sizes and likelihood approximation methods with only slight differences. With that in mind, we shall use the 3 days window and Tau leaping likelihood to get the confidence intervals for both bases since the estimates for this setting are the most consistent. In addition, the BIC for 3 days window with degree 0 basis is significantly lower than the other two.
\begin{figure}[h!]
    \centering
    \includegraphics[width=\textwidth,height=9.3cm]{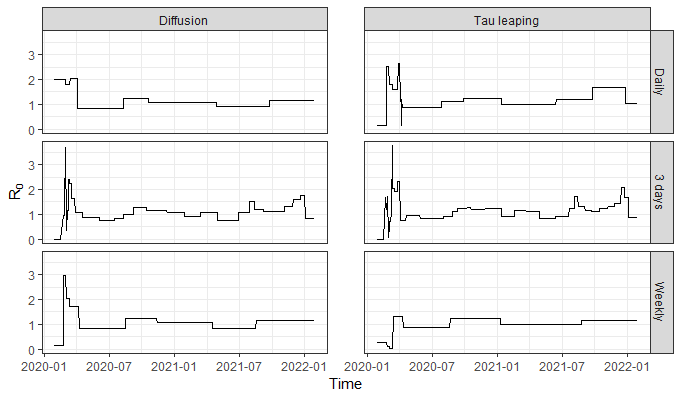}
    \caption{$R_0(t)$ estimates for COVID data using degree 0 basis}
    \label{FigCovidRes0}
\end{figure}

\begin{figure}[h!]
    \centering
    \includegraphics[width=\textwidth,height=9.3cm]{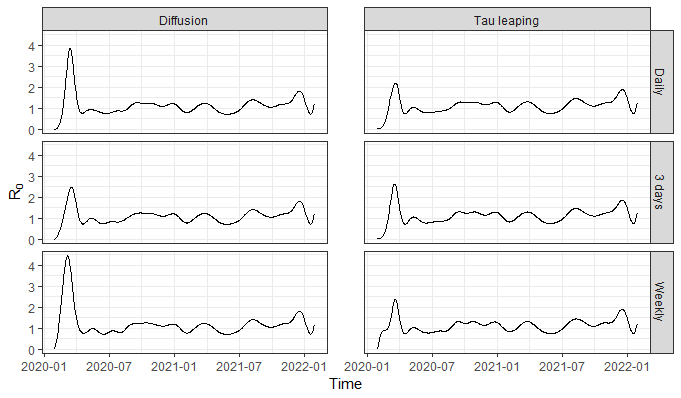}
    \caption{$R_0(t)$ estimates for COVID data using degree 3 basis}
    \label{FigCovidRes3}
\end{figure}

\begin{figure}[h!]
    \centering
    \includegraphics[width=\textwidth]{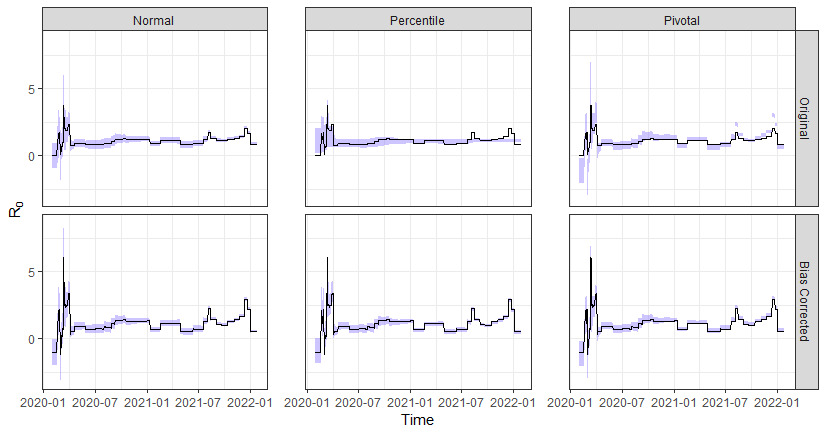}
    \caption{Confidence intervals for degree 0 basis.}
    \label{FigCovidCI0}
\end{figure}

\begin{figure}[h!]
    \centering
    \includegraphics[width=\textwidth]{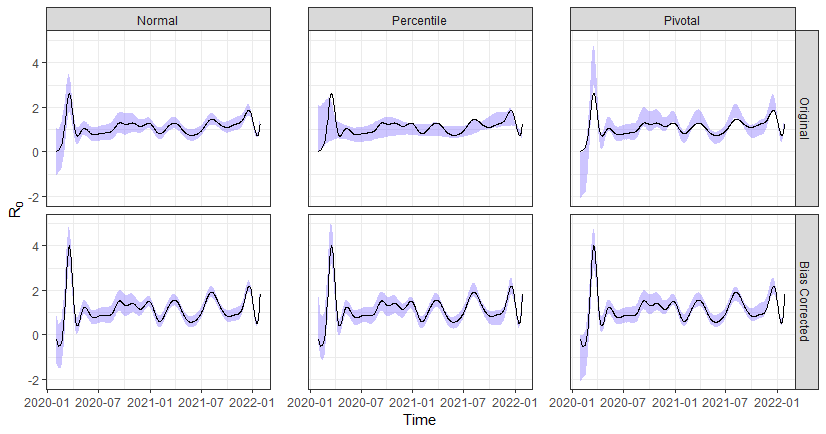}
    \caption{Confidence intervals for degree 3 basis.}
    \label{FigCovidCI3}
\end{figure}

\begin{figure}[h!]
     \centering
     \begin{subfigure}[b]{\textwidth}
         \centering
         \includegraphics[width=\textwidth,height=5.5cm]{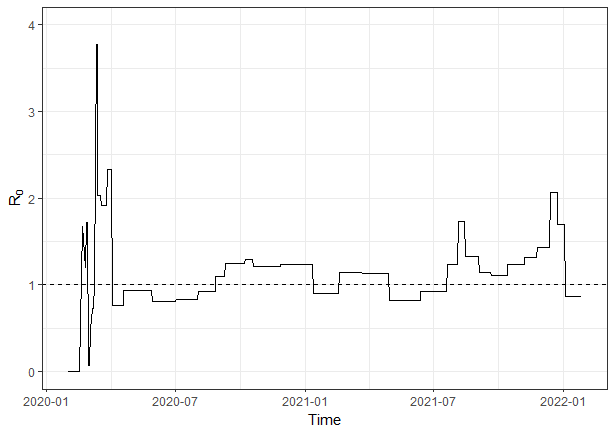}
         \caption{Degree 0 splines}
     \end{subfigure}\\
     \begin{subfigure}[b]{\textwidth}
         \centering
         \includegraphics[width=\textwidth,height=5.5cm]{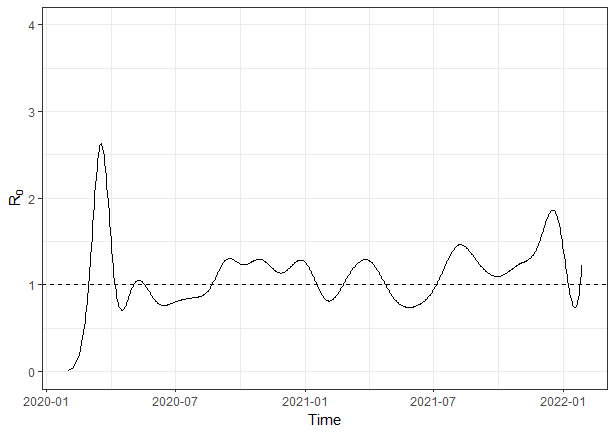}
         \caption{Degree 3 splines}
     \end{subfigure}
     \begin{subfigure}[b]{\textwidth}
         \centering
         \includegraphics[width=\textwidth,height=6.5cm]{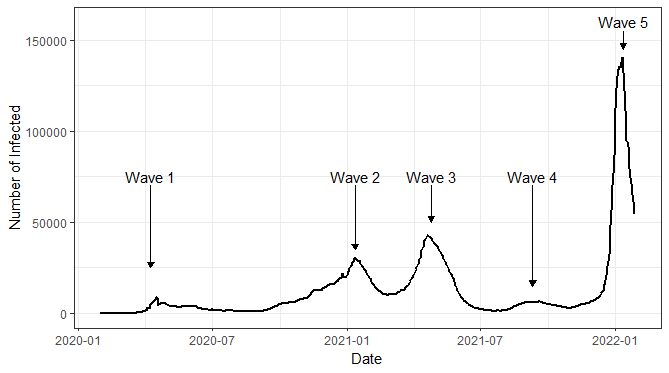}
         \caption{Active cases data}
     \end{subfigure}
        \caption{Estimated $R_0(t)$ using 3 days window compared to outbreaks.}
        \label{Fig_wave}
\end{figure}

For the confidence intervals, we use the parametric bootstrap scheme discussed in Section \ref{CI} with the samples generated by the Tau leaping method since simulation using the Gillespie algorithm is too time consuming for such a large population. 
The results are shown in Figures \ref{FigCovidCI0} and \ref{FigCovidCI3}. Note that for the pivotal interval, bias correction only applies to the estimate not the interval. Looking at the original percentile intervals, we can see that bias correction is necessary, especially for degree 0 basis. The only concern for is that the bias corrected estimate for the infection rate has a portion that lies below 0, which is clearly not true. However, the period where this happens is at the beginning of the epidemic where the number of cases is very small, which is understandable. Finally, Figure \ref{Fig_wave} shows the reproduction estimates chosen by our method compared to outbreak dates. It seems the peaks in reproduction number chosen by this model closely match the major waves.

\section{Discussion}
In this work, we develop a framework for nonparametric inference of the infection rate function for the SIR model. The two main ideas of this framework is approximating the SIR likelihood function with a different process and using a B-spline basis, which is determined by applying a knot placement method on the moving average rate estimates, to estimate the infection rate. We investigate two ways of approximating the likelihood function for the model using diffusion approximation and Tau leaping. Each of these methods can be made more accurate by using a multi-step scheme which involves simulating sample paths between observations. 

In our simulation study, the multi-step approximations have smaller biases but more variation compared to the single-step approach. 
However, the variation are relatively small for both methods. 
Therefore, the trade-off of an increase in variation for reduced bias makes a multi-step approximation preferable. 
One difficulty is that the multi-step approximation requires much greater computational resources when bootstrapping is used. Thus we expect that the reason the 20-step approximation does well is because the single-step approximation breaks down in intervals of time between observation where the process is highly non-homogeneous. We can see that these are intervals where $S(t)$ and $I(t)$ change rapidly. This suggests, as a potentially valuable future research direction, that we might reduce computational time by using an hybrid, adaptive approach with varying numbers of steps depending on the observation intervals; more steps being used in intervals where $S(t)$ and $I(t)$ are rapidly changing. In addition, if in the definition of the SIR model, the two Poisson processes were homogeneous, the Tau leaping calculations would be exact. This suggest using Tau leaping is preferable to diffusion approximation if resources allow.

In the scenario where $\beta(t)$ is constant, the single-step approach has a noticeable bias:
%There are only one direction for errors, however, which is too many segments or some curvature. 
It slightly over estimates $\beta(t)$ at early stages. Since the multi-step approximation does not show the same bias, we believe that the bias is caused by the rapid change of $S(t)$ and $I(t)$ at the early stage of the epidemic.
Not surprisingly, a degree 3 spline give better estimation when $\beta(t)$ is a continuous function and a degree 0 spline estimate is better when $\beta(t)$ is piece-wise constant. 
The degree 0 spline still perform surprisingly well for continuous $\beta(t)$ and the degree 3 spline can still capture the shape of a piece-wise constant function $\beta(t)$. It is possible that degree 3 estimation might perform better in cases with change points with different approaches to choosing knots that allow rapid changes in small intervals, which could be valuable for future research. 

The choice of a degree 0 or degree 3 spline becomes more important in constructing confidence intervals. 
Our results show that the biases in estimation creates substantial problems with confidence intervals. 
With bias correction and smoothing, the problems of low coverage could be remedied to a degree but only if the choice of spline is appropriate for the unknown transition rate $\beta(t)$ (Figure \ref{Fig_CI_final}). 
That is, when $\beta(t)$ is piece-wise constant, we should use the degree 0 spline and when $\beta(t)$ is a continuous function, the degree 3 spline is more appropriate. 
Note that, this did not completely alleviate the problems for change point settings and even the degree 0 confidence intervals tends to undercover near some of the change-points with lower magnitude of change.
Furthermore, because coverage is good at nearby time points, the practical implication in applications is that although change-points can be ascribed as having occurred with some certainty, the precise location of change may be difficult to pin down. Future work on adjustments to knot-finding and use of multi-step approximations alluded to above may give better performance and help to make the degree of the spline less important to confidence interval calculation. 

Finally, we apply our methods to the COVID-19 data in Ontario over a two year period. The final models agree with the $5$ major COVID-19 waves suggesting that our approach can be used to predict emerging epidemic waves.
Also of interest is that waves are preceded by (often relatively short) periods of time where $R_0(t)$ is larger than 1. This suggests that the approach might be used in an on-line manner, with continuous updating, as a prediction of waves. 
We can also consider more complex compartmental models, such as the SEIR (Susceptible-Exposed-Infected-Removed) model, and allowing the recovery rate $\gamma$ to vary through time for future works.

\section*{Acknowledgement}
ES and LSTH were supported by the Natural Sciences and Engineering Research Council of Canada (NSERC) Discovery
Grants and a grant from the NSERC Emerging Infectious Diseases Modelling Initiative program.

\clearpage

\bibliographystyle{plainurl}
\bibliography{references}

\begin{thebibliography}{10}

\bibitem{berge2017simple}
Tsanou Berge, JM-S Lubuma, GM~Moremedi, Neil Morris, and R~Kondera-Shava.
\newblock A simple mathematical model for {Ebola} in {Africa}.
\newblock {\em Journal of Biological Dynamics}, 11(1):42--74, 2017.

\bibitem{Blum2010}
Michael~GB Blum and Viet~Chi Tran.
\newblock {HIV} with contact tracing: a case study in approximate {Bayesian}
  computation.
\newblock {\em Biostatistics}, 11(4):644--660, 2010.

\bibitem{Britton2019}
Tom Britton, Etienne Pardoux, Franck Ball, Catherine Laredo, David Sirl, and
  Viet~Chi Tran.
\newblock {\em Stochastic epidemic models with inference}.
\newblock Springer, 2019.

\bibitem{bu2024stochastic}
Fan Bu, Allison~E Aiello, Alexander Volfovsky, and Jason Xu.
\newblock Stochastic {EM} algorithm for partially observed stochastic epidemics
  with individual heterogeneity.
\newblock {\em Biostatistics}, page kxae018, 2024.

\bibitem{Covid_data}
Public Health Agency~of Canada.
\newblock Covid-19 epidemiology update: Key updates, Jan 2023.
\newblock URL: \url{https://health-infobase.canada.ca/covid-19/}.

\bibitem{Chen2020}
Yi-Cheng Chen, Ping-En Lu, Cheng-Shang Chang, and Tzu-Hsuan Liu.
\newblock A time-dependent {SIR} model for {COVID-19} with undetectable
  infected persons.
\newblock {\em IEEE Transactions on Network Science and Engineering},
  7(4):3279--3294, 2020.

\bibitem{clemenccon2008stochastic}
St{\'e}phan Cl{\'e}men{\c{c}}on, Viet Chi~Tran, and Hector De~Arazoza.
\newblock A stochastic {SIR} model with contact-tracing: large population
  limits and statistical inference.
\newblock {\em Journal of Biological Dynamics}, 2(4):392--414, 2008.

\bibitem{Dargatz2006}
Christiane Dargatz.
\newblock A diffusion approximation for an epidemic model.
\newblock {\em Sonderforschungsbereich}, 386:517, 2006.

\bibitem{efron1994introduction}
Bradley Efron and Robert~J Tibshirani.
\newblock {\em An introduction to the bootstrap}.
\newblock CRC press, 1994.

\bibitem{Population_data}
Statistics~Canada Government~of Canada.
\newblock Population and dwelling counts: Canada, provinces and territories,
  Feb 2022.
\newblock URL:
  \url{https://www150.statcan.gc.ca/t1/tbl1/en/tv.action?pid=9810000101}.

\bibitem{ho2018direct}
Lam Si~Tung Ho, Forrest~W Crawford, and Marc~A Suchard.
\newblock Direct likelihood-based inference for discretely observed stochastic
  compartmental models of infectious disease.
\newblock {\em The Annals of Applied Statistics}, 12(3):1993--2021, 2018.

\bibitem{ho2018birth}
Lam Si~Tung Ho, Jason Xu, Forrest~W Crawford, Vladimir~N Minin, and Marc~A
  Suchard.
\newblock Birth/birth-death processes and their computable transition
  probabilities with biological applications.
\newblock {\em Journal of Mathematical Biology}, 76:911--944, 2018.

\bibitem{huang2024detecting}
Jenny Huang, Rapha{\"e}l Morsomme, David Dunson, and Jason Xu.
\newblock Detecting changes in the transmission rate of a stochastic epidemic
  model.
\newblock {\em Statistics in Medicine}, 43(10):1867--1882, 2024.

\bibitem{Kermack1927}
William~Ogilvy Kermack and Anderson~G McKendrick.
\newblock A contribution to the mathematical theory of epidemics.
\newblock {\em Proceedings of the Royal Society of London. Series A, Containing
  papers of a mathematical and physical character}, 115(772):700--721, 1927.

\bibitem{Kloeden1992}
Peter~E Kloeden and Eckhard Platen.
\newblock Higher-order implicit strong numerical schemes for stochastic
  differential equations.
\newblock {\em Journal of Statistical Physics}, 66:283--314, 1992.

\bibitem{kuhnert2014simultaneous}
Denise K{\"u}hnert, Tanja Stadler, Timothy~G Vaughan, and Alexei~J Drummond.
\newblock Simultaneous reconstruction of evolutionary history and
  epidemiological dynamics from viral sequences with the birth--death {SIR}
  model.
\newblock {\em Journal of the Royal Society Interface}, 11(94):20131106, 2014.

\bibitem{Osthus2017}
Dave Osthus, Kyle~S Hickmann, Petru{\c{t}}a~C Caragea, Dave Higdon, and Sara~Y
  Del~Valle.
\newblock Forecasting seasonal influenza with a state-space {SIR} model.
\newblock {\em The Annals of Applied Statistics}, 11(1):202, 2017.

\bibitem{raggett1982stochastic}
GF~Raggett.
\newblock A stochastic model of the eyam plague.
\newblock {\em Journal of Applied Statistics}, 9(2):212--225, 1982.

\bibitem{Roda2020}
Weston~C Roda, Marie~B Varughese, Donglin Han, and Michael~Y Li.
\newblock Why is it difficult to accurately predict the {COVID-19} epidemic?
\newblock {\em Infectious Disease Modelling}, 5:271--281, 2020.

\bibitem{saha2023spade4}
Esha Saha, Lam Si~Tung Ho, and Giang Tran.
\newblock {SPADE4}: Sparsity and delay embedding based forecasting of
  epidemics.
\newblock {\em Bulletin of Mathematical Biology}, 85(8):71, 2023.

\bibitem{Smirnova2019}
Alexandra Smirnova, Linda deCamp, and Gerardo Chowell.
\newblock Forecasting epidemics through nonparametric estimation of
  time-dependent transmission rates using the {SEIR} model.
\newblock {\em Bulletin of Mathematical Biology}, 81:4343--4365, 2019.

\bibitem{villani2009optimal}
C{\'e}dric Villani et~al.
\newblock {\em Optimal transport: old and new}, volume 338.
\newblock Springer, 2009.

\bibitem{Wasserman2006}
Larry Wasserman.
\newblock {\em All of nonparametric statistics}.
\newblock Springer Science \& Business Media, 2006.

\bibitem{watson2021pandemic}
Gregory~L Watson, Di~Xiong, Lu~Zhang, Joseph~A Zoller, John Shamshoian, Phillip
  Sundin, Teresa Bufford, Anne~W Rimoin, Marc~A Suchard, and Christina~M
  Ramirez.
\newblock {Pandemic velocity: Forecasting COVID-19 in the US with a machine
  learning \& Bayesian time series compartmental model}.
\newblock {\em PLoS Computational Biology}, 17(3):e1008837, 2021.

\bibitem{whittles2016epidemiological}
Lilith~K Whittles and Xavier Didelot.
\newblock Epidemiological analysis of the {Eyam} plague outbreak of 1665--1666.
\newblock {\em Proceedings of the Royal Society B: Biological Sciences},
  283(1830):20160618, 2016.

\bibitem{yang2015forecasting}
Wan Yang, Benjamin~J Cowling, Eric~HY Lau, and Jeffrey Shaman.
\newblock Forecasting influenza epidemics in {Hong Kong}.
\newblock {\em PLoS Computational Biology}, 11(7):e1004383, 2015.

\bibitem{Yeh2020}
Raine Yeh, Youssef~S.G. Nashed, Tom Peterka, and Xavier Tricoche.
\newblock Fast automatic knot placement method for accurate {B-spline} curve
  fitting.
\newblock {\em Computer-Aided Design}, 128:102905, 2020.

\bibitem{ssar}
Rodrigo Zepeda and Dalia Camacho.
\newblock {\em ssar: A speedy implementation of Gillespie's Stochastic
  Simulation Algorithm}.
\newblock R package version 0.0.0.9000.

\end{thebibliography}

\appendix
\section{Proof of lemma \ref{lem_Wasserstein}}
For all $t\in [0,T],\epsilon>0$ we have
\begin{equation}\label{2.11}
\begin{alignedat}{3}
E(\|U_n-V\|_T)&
=E(\sup\limits_{t\in [0,T]}\|U_n(t)-V(t)\|)
\ge E(\|U_n(t)-V(t)\|)\\
&\ge \epsilon P(\|U_n(t)-V(t)\|>\epsilon).
\end{alignedat}
\end{equation}

Taking infimum over all couplings of $U_n(t)$ and $V(t)$ in \eqref{2.11} gives
\begin{align*}
    W_{1,T}(U_n,V)\ge \epsilon\inf P(\|U_n(t)-V(t)\|>\epsilon).
\end{align*}

Since $W_{1,T}(U_n,V)\xrightarrow{n\to\infty} 0$, we have
\begin{align}\label{2.13}
    \inf P(\|U_n(t)-V(t)\|>\epsilon)\xrightarrow{n\to\infty} 0\quad \forall\ \epsilon>0.
\end{align}

Next, we have for all $u\in \mathbb{R}^d, \epsilon>0$
\begin{equation}\label{2.14}
\begin{alignedat}{3}
F_{U_n(t)}(u)&=P(U_n(t)\le u)\le P(V(t)\le u+\epsilon\textbf{1})+P(\|U_n(t)-V(t)\|>\epsilon)\\
&=F_{V(t)}(u+\epsilon\textbf{1})+P(\|U_n(t)-V(t)\|>\epsilon)
\end{alignedat}
\end{equation}
where $\textbf{1}$ is the vector of 1's and the inequalities here are element-wise.
This is true since if $U_n(t)\le u$ and $\|U_n(t)-V(t)\|\le\epsilon$ then $V_n\le u+\epsilon\textbf{1}$. Applying this for $u-\epsilon\textbf{1}$ with the role of $U_n(t)$ and $V(t)$ swapped, we have
\begin{align}\label{2.15}
F_{V(t)}(u-\epsilon\textbf{1})&\le F_{U_n(t)}(u)+P(\|U_n(t)-V(t)\|>\epsilon).
\end{align}

Combining \eqref{2.14} and \eqref{2.15} gives us
\begin{equation}\label{2.16}
\begin{alignedat}{3}
F_{V(t)}(u-\epsilon\textbf{1})-P(\|U_n(t)-V(t)\|>\epsilon)
\le F_{U_n(t)}(u) \\
\le F_{V(t)}(u+\epsilon\textbf{1})+P(\|U_n(t)-V(t)\|>\epsilon).
\end{alignedat}
\end{equation}

In \eqref{2.16}, taking the infimum over all couplings of $U_n(t)$ and $V(t)$ gives
\begin{equation}\label{2.17}
\begin{alignedat}{3}
F_{V(t)}(u-\epsilon\textbf{1})-\inf P(\|U_n(t)-V(t)\|>\epsilon)
\le F_{U_n(t)}(u) \\
\le F_{V(t)}(u+\epsilon\textbf{1})+\inf P(\|U_n(t)-V(t)\|>\epsilon).
\end{alignedat}
\end{equation}

Note that the cdf's are not affected by the coupling since the marginals are fixed. This combined with \eqref{2.13} and letting $\epsilon\to 0,\ n\to\infty$ gives us
\begin{align}
F_{U_n(t)}(u)\xrightarrow{n\to\infty} F_{V(t)}(u).
\end{align}
In other words, $U_n(t) \xrightarrow{d} V(t)$ as $n\to\infty$ for all $t\in [0,T]$.

\end{document}